\documentclass[12pt,a4paper,reqno]{amsart}
\usepackage{amsmath, amsfonts,
   amssymb, amsbsy, amscd,
   amsthm} 

\newtheorem{theor}{Theorem}
\theoremstyle{definition}

\newtheorem{state}[theor]{Proposition}

\newtheorem{example}{Example}

\theoremstyle{remark}
\newtheorem{rem}{Remark}

\DeclareFontFamily{OML}{cyr}{}
\DeclareFontShape{OML}{cyr}{m}{n}{
   <5> <6> <7> <8> <9> gen * wncyr
   <10> <10.95> <12> <14.4> <17.28> <20.74> <24.88> wncyr10
  }{}
\DeclareSymbolFont{rusletters}{OML}{cyr}{m}{n}
\DeclareSymbolFontAlphabet{\rusmath}{rusletters}
\DeclareMathSymbol\re{\rusmath}{rusletters}{"03}

\newcommand{\pinner}{\mathbin{\mathchoice
   {\hbox{\vrule width0.6em depth0pt height0.4pt
   \vrule width0.4pt depth0pt height0.8ex}}
   {\hbox{\vrule width0.6em depth0pt height0.4pt
   \vrule width0.4pt depth0pt height0.8ex}}
   {\hbox{\kern0.14em
   \vrule width0.48em depth0pt height0.4pt
   \vrule width0.4pt depth0pt height0.6ex\kern0.14em}}
   {\hbox{\kern0.1em
   \vrule width0.39em depth0pt height0.4pt
   \vrule width0.4pt depth0pt height0.5ex\kern0.1em}}}}

\newcommand{\BBR}{\mathbb{R}}
\newcommand{\BBN}{\mathbb{N}}

\newcommand{\cD}{\mathcal{D}}
\newcommand{\cE}{\mathcal{E}}

\newcommand{\cF}{\mathcal{F}}
\newcommand{\cH}{\mathcal{H}}

\newcommand{\cL}{\mathcal{L}}

\newcommand{\bh}{{\boldsymbol{h}}}

\newcommand{\bu}{\boldsymbol{u}}

\newcommand{\gm}{\mathfrak{m}}

\newcommand{\dd}{\partial}
\newcommand{\Id}{{\mathrm d}}

\newcommand{\rmi}{{\mathrm i}}

\DeclareMathOperator{\img}{im}

\DeclareMathOperator{\Sres}{Sres}

\newcommand{\tu}{\tilde{u}}

\newcommand{\by}[1]{\textit{{#1}}}
\newcommand{\jour}[1]{\textit{{#1}}}
\newcommand{\vol}[1]{\textbf{{#1}}}
\newcommand{\book}[1]{\textrm{{#1}}}

\title[Gardner's deformations of $N{=}2$, $a{=}4$ super\/-\/KdV]%
{Gardner's deformations of the 
$N{=}2$ supersymmetric $a{=}4$--\/KdV equation} 

\date{November~11, 2009}

\author[V.~Hussin]{V. 
Hussin}
\address{D\'epartement de Math\'ematiques et de Statistique,
Universit\'e de Montr\'eal,
C.P.~6128, succ.\ Centre\/-\/ville, Montr\'eal, Qu\'ebec H3C~3J7,
Canada.}
\curraddr{
Department of Mathematical Sciences, Durham University, 
Science Laboratories, South~Rd., Durham DH1~3LE, United Kingdom.%
}
\email{
hussin
@dms.umontreal.ca}

\author[A. V. Kiselev]{A. 
V. Kiselev${}^{*}$} 
\thanks{${}^*$\textit{Address for correspondence}:
   Max Planck Institute for Mathematics,
   Vivatsgasse 7, D-53111 Bonn, Germany.
\textit{E-mail}: \texttt{arthemy\symbol{"40}moim-bonn.mpg.de}.}
\address{
Mathematical Institute, University of Utrecht, P.O.Box~80.010, 3508~TA
Utrecht, The Netherlands.}
\email{A.V.Kiselev@uu.nl}

\author[A.~O.~Krutov]{A. 
O.~Krutov}
\address{%
Department of Higher Mathematics, Ivanovo State Power 
University, 34\,Rabfa\-kovskaya str., Ivanovo, 153003 Russia.}
\email{krutov@math.ispu.ru} 

\author[T.~Wolf]{T. 
Wolf}
\address{Department of Mathematics, Brock University,
500 Glenridge av.,  St.~Catharines, Ontario L2S~3A1, Canada.}
\email{twolf@brocku.ca}

\dedicatory{This paper is an extended version of the talks given by
A.~V.~K. at the 
8th International Conference `Symmetry in Nonlinear Mathematical Physics'
(June~20--27, 2009), Kiev, Ukraine, and at the 
9th International Workshop SQS'09
`Supersymmetry and quantum symmetries' (July~29 -- August~3,
2009), JINR, Dubna, Russia.}



\begin{document}
\rightline{\jour{Preprint} MPIM(Bonn)-???/2009}

\begin{abstract}
We prove that P.~Mathieu's Open problem on constructing Gardner's deformation for   
the $N{=}2$ supersymmetric $a{=}4$--\/Korteweg\/--\/de Vries equation
has no supersymmetry\/-\/invariant solutions, whenever it is assumed that 
they retract to Gardner's deformation of the scalar KdV equation under the 
component reduction. 
At the same time, we propose a two\/-\/step scheme for the recursive
production of the integrals of motion for the $N{=}2$,\ $a{=}4$--\/SKdV.
First, we find a new Gardner's deformation of the Kaup\/--\/Boussinesq equation, which is contained in the bosonic limit of 
the super\/-\/
hierarchy. This yields the recurrence relation between the Hamiltonians of the limit, whence we determine the bosonic super\/-\/Hamiltonians of the full
$N{=}2$, $a{=}4$--\/SKdV hierarchy. 
Our method is applicable towards the solution of Gardner's deformation problems for other supersymmetric KdV\/-\/type systems.
\end{abstract}
\maketitle

\subsection*{Introduction}
This paper is devoted to the Korteweg\/--\/de Vries equation and its generalizations~\cite{Gardner}.
We consider completely integrable, multi\/-\/Hamiltonian
evolutionary $N{=}2$ supersymmetric equations
upon a scalar, complex bosonic $N{=}2$ superfield 
\begin{equation}\label{N2SuperField}
\bu(x,t;\theta_1,\theta_2)
=u_0(x,t)+\theta_1\cdot u_1(x,t)+\theta_2\cdot u_2(x,t)+\theta_1\theta_2
 \cdot u_{12}(x,t),
\end{equation}
where $\theta_1$ and~$\theta_2$ are Grassmann variables satisfying
$\theta_1^2=\theta_2^2=\theta_1\theta_2+\theta_2\theta_1=0$.
Also, we investigate one-{} and two\/-\/component reductions of such four\/-\/component $N{=}2$ super\/-\/systems upon~$\bu$. In particular, we study the 
bosonic limits, which are obtained by the constraint
\begin{equation}\label{BRed}
u_1=u_2\equiv0.
\end{equation}
We analyse the structures that are inherited by the limits from the full super\/-\/systems and, conversely, 
recover the integrability properties of the entire $N{=}2$ hierarchies from their bosonic counterparts.

We address 
$2^{\text{nd}}$ Open
problem of~\cite{MathieuOpen} 
for the $N{=}2$ supersymmetric Korteweg\/--\/de Vries equation with
$a{=}4$, see~\cite{MathieuTwo,MathieuNew}, 
\begin{equation}\label{SKdV}
\bu_t=-\bu_{xxx}+3\bigl(\bu\cD_1\cD_2\bu\bigr)_x
 +\frac{a-1}{2}\bigl(\cD_1\cD_2\bu^2\bigr)_x + 3a\bu^2\bu_x,\qquad
 \cD_i=\frac{\dd}{\dd\theta_i}+\theta_i\cdot\frac{\Id}{\Id x}.
\end{equation}
For $a{=}4$, this 
super\/-\/
equation possesses an infinite
hierarchy of bosonic Hamiltonian super\/-\/functionals~$\boldsymbol{\cH}^{(k)}$
whose densities~$\boldsymbol{h}^{(k)}$ 
are integrals of motion. 
The problem amounts to a recursive production of such densities by using those which are already obtained. In its authentic formulation,
the problem suggests finding a 
parametric family of super\/-\/equations~$\cE(\epsilon)$ 
upon the generating super\/-\/function 
$\tilde{\bu}(\epsilon)=\sum_{k=0}^{+\infty}\boldsymbol{h}^{(k)}\cdot\epsilon^k$
for the integrals of motion such that the initial super\/-\/equation~\eqref{SKdV} is~$\cE(0)$. 
It is further supposed that, at each~$\epsilon$, the evolutionary equation~$\cE(\epsilon)$ exprimes a (super-)\/conserved current, and there is the Gardner\/--\/Miura substitution $\gm_\epsilon\colon\cE(\epsilon)\to\cE(0)$. Hence, expanding~$\gm_\epsilon$ in~$\epsilon$ and using the initial condition 
$\tilde{\bu}(0)=\bu$ at~$\epsilon=0$, 
one obtains the differential recurrence relation between the Taylor coefficients~$\boldsymbol{h}^{(k)}$ of the generating function~$\tilde{\bu}$ 
(see~\cite{Gardner} or~\cite{PamukKale,TMPh2006,KuperIrish,MathieuNew} 
and references therein for 
details and examples).
The recurrence relations between the (super-)\/Hamiltonians of the hierarchy
are much more informative than the usual recursion operators that propagate symmetries. In particular, the symmetries can be used to produce new explicit solutions from known ones, but the integrals of motion 
help to find those primary solutions.

Let us also note that, 
within the Lax framework of super\/-\/pseudodifferential operators,
the calculation of the $(n+1)$-\/st residue does not take into account
the $n$~residues, which are already known at smaller indices.
This is why the method of Gardner's deformations becomes highly preferrable.
Indeed, there is no need to multiply any pseudodifferential operators 
by applying the Leibnitz rule an increasing number of times,
and all the previously obtained quantities are used at each inductive step.
By this argument, we understand Gardner's deformations as the transformation in the space of the integrals of motion that maps the residues to Taylor coefficients of the generating functions~$\tilde{\bu}(\epsilon)$ and which, therefore, endows this space with the additional structure (that is, with the recurrence relations between the integrals).

Still there is a deep intrinsic relation between the Lax (or, more
generally, zero\/-\/curvature) representations for integrable systems and 
Gardner's deformations for them. Namely, both approaches manifest the matrix
and vector field representations of 
the Lie algebras related to such
systems~\cite{WilsonEquiv}.

\smallskip
Our main result is the following. 
Under some natural assumptions,
we prove 
the non\/-\/existence of $N{=}2$ supersymmetry\/-\/invariant Gardner's deformations for the bi\/-\/Hamiltonian $N{=}2$,\ $a{=}4$--\/SKdV.
Still, we show that the Open problem must be addressed in a different way,
and then we solve it in two steps. 
First, in section~\ref{SecKBous} we 
recall that the tri\/-\/Hamiltonian hierarchy for
the bosonic limit of~\eqref{SKdV} with $a{=}4$ contains the
Kaup\/--\/Boussinesq equation, see~\cite{Kaup75,NutkuPavlov}
and~\cite{BrunelliDas94,KuperSuperLongWaves,PalitChowdhury96} in the 
context of this paper.
Then in section~\ref{SecBurg} we construct new 
deformations for the Kaup\/--\/Boussinesq
equation such that the Miura contraction~$\gm_\epsilon$ now incorporates
Gardner's map for the KdV equation (\cite{Gardner}, c.f.~\cite{PamukKale,KuperIrish}).
Second, extending the Hamiltonians~$H^{(k)}$ 
for the Kaup\/--\/Boussinesq hierarchy
to the super\/-\/functionals~$\boldsymbol{\cH}^{(k)}$ in section~\ref{SecHam},
we reproduce the bosonic conservation laws for~\eqref{SKdV} with~$a{=}4$.
Finally, we contribute to the solution of P.~Mathieu's $3^{\text{rd}}$ Open problem~\cite{MathieuOpen} with the description of necessary conditions upon a class of Gardner's deformations for~\eqref{SKdV} that reproduce its \emph{fermionic} local conserved densities.

\smallskip
The standard reference in geometry of completely integrable Hamiltonian partial differential equations is~\cite{Olver}.

\section{$N{=}2$ $a{=}4$--\/SKdV as bi\/-\/Hamiltonian super\/-\/extension of Kaup\/--\/Boussinesq system}\label{SecKBous}
\noindent%
Let us begin with the Korteweg\/--\/de Vries equation 
\begin{equation}\label{KdV}
u_{12;t}+u_{12;xxx}+6u_{12}u_{12;x}=0.
\end{equation}
Its 
second Hamiltonian operator, 
$\smash{\hat{A}_2^{\text{KdV}}} = \Id^3/\Id x^3+4u_{12}\,\Id/\Id x
+2u_{12;x}$,
which relates~\eqref{KdV} to the functional 
$H^{(2)}_{\text{KdV}}=-\tfrac{1}{2}\int u_{12}^2\,\Id x$, can be
extended%
\footnote{Likewise, 
   we will extend Gardner's deformation~\eqref{DefKdV}
   of~\eqref{KdV} to the deformation~\eqref{AppR}
   of the two\/-\/component bosonic limit~\eqref{BLim} for~\eqref{SKdV}
   with $a{=}4$. 
      Hence we reproduce the conservation laws
   for~\eqref{BLim} and, again, 
   extend them to the bosonic super\/-\/%
   Hamiltonians of the full 
   system~\eqref{SKdV}.}
in the $(2\mid2)$-\/graded field setup 
to the parity\/-\/preserving 
Hamiltonian 
operator~\cite{MathieuTwo},
\begin{equation}\label{SecondHam4x4}
\hat{P}_2=\begin{pmatrix}
-\tfrac{\Id}{\Id x} & -u_2 & u_1 & 2u_0\tfrac{\Id}{\Id x}+2u_{0;x} \\
-u_2 & \bigl(\tfrac{\Id}{\Id x}\bigr)^2+u_{12} & -2u_0\tfrac{\Id}{\Id x}-u_{0;x} & 3u_1\tfrac{\Id}{\Id x}+2u_{1;x} \\
u_1 & 2u_0\tfrac{\Id}{\Id x}+u_{0;x} & \bigl(\tfrac{\Id}{\Id x}\bigr)^2+u_{12}\vphantom{\Bigr)} & 3u_2\tfrac{\Id}{\Id x}+2u_{2;x} \\
2u_0\tfrac{\Id}{\Id x} & -3u_1\tfrac{\Id}{\Id x}-u_{1;x} & -3u_2\tfrac{\Id}{\Id x}-u_{2;x} & \underline{\bigl(\tfrac{\Id}{\Id x}\bigr)^3+4u_{12}\tfrac{\Id}{\Id x}+2u_{12;x}}
\end{pmatrix}.
\end{equation}
Here the fields $u_0$ and~$u_{12}$ are bosonic,
$u_1$~and $u_2$~are fermionic together with their derivatives w.r.t.~$x$.
Likewise, the components $\psi_0\simeq\delta\cH/\delta u_0$ and
$\psi_{12}\simeq\delta\cH/\delta u_{12}$ of the arguments $\vec{\psi}={}^t\bigl(\psi_0,\psi_1,\psi_2,\psi_{12}\bigr)$ of~\eqref{SecondHam4x4} are even\/-\/graded and $\psi_1$, $\psi_2$~are odd\/-\/graded.
The operator~\eqref{SecondHam4x4} is unique in the class of Hamiltonian total differential operators that merge to scalar $N{=}2$ super\/-\/operators 
which are local in $\cD_i$ and whose
coefficients depend on the super\/-\/field~$\bu$ and its super\/-\/derivatives, see~\eqref{SecondHamN=2} below.
The operator~
\eqref{SecondHam4x4} determines 
the 
$N{=}2$ classical super\/-\/conformal algebra~\cite{ChaichianKulish87}.
Conversely, the Poisson bracket given by~\eqref{SecondHam4x4}
reduces to the second Poisson bracket
for~\eqref{KdV}, whenever one sets equal to zero 
the fields~$u_0$,\ $u_1$,\ and~$u_2$ both in the coefficients of~\eqref{SecondHam4x4} and in all Hamiltonians; the operator $\hat{A}_2^{\text{KdV}}$ is underlined in~\eqref{SecondHam4x4}.

By construction, P.~Mathieu's extensions of the Korteweg\/--\/de Vries equation~\eqref{KdV} are determined by the operator~\eqref{SecondHam4x4} and the bosonic Hamiltonian functional
\begin{equation}\label{HamComponents}
\cH^{(2)}=\int\Bigl[u_0u_{0;xx}\underline{-u_{12}^2}+u_1u_{1;x}+u_2u_{2;x}+
   a\cdot\bigl(u_0^2u_{12}-2u_0u_1u_2\bigr)\Bigr]\,\Id x,
\end{equation}
which incorporates $H^{(2)}_{\text{KdV}}$ 
as the underlined term; similar to~\eqref{SecondHamN=2}, the Hamiltonian~\eqref{HamComponents} will be realized by~\eqref{BalanceHam} as the bosonic $N{=}2$ super\/-\/Hamiltonian. 
Now we have that
\[
u_{i;t}=\bigl(\hat{P}_2\bigr)_{ij}\bigl(\delta\cH^{(2)}/\delta u_j\bigr),
\qquad i,j\in\{0,1,2,12\}.
\]
This yields the system
\begin{subequations}\label{SKdVComponents}
\begin{align}
u_{0;t}&=-u_{0;xxx}+\bigl(a u_0^3
   -(a+2)u_0u_{12}+(a-1)u_1u_2\bigr)_x,
\label{GetmKdV}\\
u_{1;t}&=-u_{1;xxx}+\bigl(\phantom{+}(a+2)u_0u_{2;x}+(a-1)u_{0;x}u_2
   -3u_1u_{12}+3a u_0^2u_1 \bigr)_x,\\
u_{2;t}&=-u_{2;xxx}+\bigl(-(a+2)u_0u_{1;x}-(a-1)u_{0;x}u_1
   -3u_2u_{12}+3a u_0^2u_2 \bigr)_x,\\
\underline{u_{12;t}}&=\underline{-u_{12;xxx}-6u_{12}u_{12;x}}
 +3au_{0;x}u_{0;xx}+(a+2)u_0u_{0;xxx}\notag\\
 {}&{}\qquad{}+3u_1u_{1;xx}+3u_2u_{2;xx}
 +3a\bigl(u_0^2u_{12} -2u_0u_1u_2\bigr)_x.\label{GetKdV}
\end{align}
\end{subequations}
Obviously, it retracts to~\eqref{KdV}, which we underline in~\eqref{SKdVComponents}, under the reduction $u_0=0$,~$u_1=u_2=0$.

At all~$a\in\BBR$, the Hamiltonian~\eqref{HamComponents} equals
\begin{equation}\label{BalanceHam}
\boldsymbol{\cH}^{(2)}=\int\bigl(
   \boldsymbol{u}\cD_1\cD_2(\boldsymbol{u})+\tfrac{a}{3}\boldsymbol{u}^3
\bigr)\Id\boldsymbol{\theta}\Id x,\qquad 
\text{where $\Id\boldsymbol{\theta}=\Id\theta_1\Id\theta_2$.}
\end{equation}
Likewise, the structure~\eqref{SecondHam4x4}, which is independent of~$a$, 
produces the $N{=}2$ super\/-\/operator
\begin{equation}\label{SecondHamN=2}
\hat{\boldsymbol{P}}_2=
\cD_1\cD_2\tfrac{\Id}{\Id x}+2\bu\tfrac{\Id}{\Id x}-\cD_1(\bu)\cD_1-\cD_2(\bu)\cD_2+2\bu_x.
\end{equation}
Thus we recover P.~Mathieu's super\/-\/equations~\eqref{SKdV}~\cite{MathieuNew}, which are Hamiltonian with respect to~\eqref{SecondHamN=2}
and the functional~\eqref{BalanceHam}: 
$\bu_t=\hat{\boldsymbol{P}}_2\bigl(\tfrac{\delta}{\delta\bu}
 (\boldsymbol{\cH}_2)\bigr)$. 
In component notation, super\/-\/equations~\eqref{SKdV} are~\eqref{SKdVComponents}.

The assumption that, for a given~$a$, the super\/-\/system~\eqref{SKdV}
admits infinitely many integrals of motion
yields the triplet $a\in\{-2,1,4\}$, see~\cite{MathieuNew}.
The same values of~$a$ 
are exhibited by the Painlev\'e analysis for $N{=}2$ super\/-\/equations~\eqref{SKdV}, see~\cite{BourqueMathieu}.

The three systems~\eqref{SKdV} have the common second Poisson structure, 
which is given by~\eqref{SecondHamN=2}, but the three `junior' first
Hamiltonian operators~$\hat{\boldsymbol{P}}_1$ for them 
do not coincide~\cite{KerstenBiHamA1,MathieuTwo,MathieuNew}. 
Moreover, system~\eqref{SKdV} with $a{=}4$ is radically 
different from the other two, both from the Hamiltonian and Lax viewpoints.

\begin{state}
The $N{=}2$ supersymmetric hierarchy of P.~Mathieu's
$a{=}4$ Korteweg\/--\/de Vries equation is bi\/-\/Hamiltonian
with respect to the local super\/-\/operator~\eqref{SecondHamN=2} and
the junior Hamiltonian operator\footnote{The nonzero entries of the
$(4\times4)$-\/matrix representation $\hat{P}_1$ for 
the Hamiltonian super\/-\/operator $\hat{\boldsymbol{P}}_1^{a{=}4}$ 
are $\bigl(\hat{P}_1\bigr)_{0,12}=\bigl(\hat{P}_1\bigr)_{2,1}
 =\bigl(\hat{P}_1\bigr)_{12,0}=-\bigl(\hat{P}_1\bigr)_{1,2}=\Id/\Id x$.} 
$\hat{\boldsymbol{P}}{\mathstrut}_1^{a=4}={\Id}/{\Id x}$,
which is obtained from~$\hat{\boldsymbol{P}}{\mathstrut}_2^{a=4}$ by the shift
$\bu\mapsto\bu+\boldsymbol{\lambda}$ of the super\/-\/field~$\bu$, 
see~\textup{\cite{Magri2000-5,Artur}}:
\[
\hat{\boldsymbol{P}}{\mathstrut}_1^{a=4}=\frac{\Id}{\Id x}=\frac{1}{2}\cdot
 \frac{\Id}{\Id\boldsymbol{\lambda}}{\Bigr|}_{\boldsymbol{\lambda}=0}
  \hat{\boldsymbol{P}}{\mathstrut}_2^{a=4}{\Bigr|}_{\bu+\boldsymbol{\lambda}}.
\]   
The two operators are Poisson compatible and generate the tower of \emph{nonlocal} higher structures $\hat{\boldsymbol{P}}_{k+2}=\bigl(\hat{\boldsymbol{P}}_2\circ\hat{\boldsymbol{P}}_1^{-1}\bigr)^{k}\circ\hat{\boldsymbol{P}}_2$, $k\geq1$, for the $N{=}2$,\ $a{=}4$--\/SKdV hierarchy, see~\cite{YKSMagri,JKGolovko2008}. Although $\hat{\boldsymbol{P}}_3$ is nonlocal (c.f.~\cite{PalitChowdhury96}), 
its bosonic limits under~\eqref{BRed} yield the \emph{local} third Hamiltonian structure~$\hat{A}_{2}$ for the Kaup\/--\/Boussinesq equation, which determines the evolution along the second time~$t_2\equiv\xi$
in the bosonic limit of the $N{=}2$,\ $a{=}4$--\/SKdV hierarchy 
\textup{(}see Proposition~\textup{\ref{ThTriHam}} 
on p.~\textup{\pageref{ThTriHam})}.
\end{state}

\begin{rem}
The Kaup\/--\/Boussinesq system~\cite{Kaup75} arising here is equivalent to the Kaup\/--\/Broer system (the difference amounts to notation). 
A bi\/-\/Hamiltonian $N{=}2$ super\/-\/extension of the latter is known 
from~\cite{KuperSuperLongWaves}.
A tri\/-\/Hamiltonian two\/-\/fermion $N{=}1$ super\/-\/extension of 
the Kaup\/-\/Broer system was constructed in~\cite{BrunelliDas94} such that
in the bosonic limit the three known Hamiltonian structures for the 
initial system are recovered.    
At the same time, a boson\/-\/fermion $N{=}1$ super\/-\/extension of the Kaup\/--\/Broer equation with two local and the nonlocal third Hamiltonian structures
was derived in~\cite{PalitChowdhury96}; seemingly, the latter 
equaled the composition 
$\hat{\boldsymbol{P}}_2\circ\hat{\boldsymbol{P}}_1^{-1}\circ
\hat{\boldsymbol{P}}_2$, but 
it 
remained to prove that the suggested nonlocal super\/-\/operator is skew\/-\/adjoint, that the bracket induced on the space of bosonic super\/-\/Hamiltonians does satisfy the Jacobi identity, and that the hierarchy flows produced by the nonlocal operator remain local.
\end{rem}

There is a deep reason for the geometry of the $a{=}4$--\/SKdV 
to be exceptionally rich. 
All the three integrable $N{=}2$ supersymmetric KdV equations~\eqref{SKdV}
admit the Lax representations 
   $L_{t_3}=[A^{(3)}
,L]$, see~\cite{BonoraKrivonosSorin,MathieuTwo,MathieuOpen,PopowiczLax}.
For $a{=}4$, 
the four roots of the Lax operator~$L_{a{=}4}=-(\cD_1\cD_2+\bu)^2$, 
which are $\cL_{1,\pm}=\pm\rmi(\cD_1\cD_2+\bu)$, $\rmi^2=-1$,
and the super\/-\/pseudodifferential operators $\cL_{2,\pm}=\pm\tfrac{\Id}{\Id x}
+ 
\sum_{i>0} (\cdots)\cdot\bigl(\tfrac{\Id}{\Id x}\bigr)^{-i}$,
generate the odd\/-\/index flows of the SKdV hierarchy via 
$L_{t_{2k+1}}=[(\cL_2^{2k+1})_{\geq0},L]$.
In particular, 
we have $A_{a{=}4}^{(3)}=\bigl(L^{3/2}\bigr)_{\geq0}\mod(\cD_1\cD_2+\bu)^3$.
However, 
the \emph{entire} $a{=}4$
hierarchy is reproduced in the Lax form via $(\cL_1^k\cL_2)_{t_\ell}=
 \bigl[\bigl(\cL_1^\ell\cL_2\bigr)_{\geq0}, \cL_1^k\cL_2\bigr]$ 
for all~$k\in\BBN$, c.f.~\cite{KrivonosSorinToppan}. 
Hence the super\/-\/residues\footnote{We recall that the $N{=}2$
super\/-\/residue $\Sres M$ of a super\/-\/pseudodifferential operator~$M$
is the coefficient of $\cD_1\cD_2\circ\bigl(\tfrac{\Id}{\Id x}\bigr)^{-1}$
in~$M$%
.}
of the operators $\cL_1^k\cL_2$ are conserved. 

Consequently, unlike the other two, 
super\/-\/equation~\eqref{SKdV} with $a{=}4$ admits twice as many
constants of motion as there are for the 
super\/-\/equations with $a{=}-2$ or~$a{=}1$. 
For convenience, let us recall that super\/-\/equations~\eqref{SKdV}
are homogeneous with respect to the weights $|\Id/\Id x|\equiv1$, $|\bu|=1$,
$|\Id/\Id t|=3$. Hence we conclude that, 
for each nonnegative integer~$k$, there appears the
nontrivial conserved density~$\Sres\cL_1^k\cL_2$, see above,
of weight~$k+1$. The even weights also enter the play. 
Consequently, there are twice as many commuting super\/-\/flows
assigned to the twice as many Hamiltonians.

\begin{example}
The additional 
super\/-\/Hamiltonian 
$\boldsymbol{\cH}^{(1)}=\tfrac{1}{2}\int\bu^2\Id\boldsymbol{\theta}\Id x$
for~\eqref{SKdV} with $a{=}4$, and the second structure~\eqref{SecondHamN=2},
---~or, equivalently, the first operator $\hat{\boldsymbol{P}}_1=\Id/\Id x$
and the Hamiltonian~$\boldsymbol{\cH}^{(2)}$, or~$\hat{\boldsymbol{P}}_3$ and 
$\boldsymbol{\cH}^{(0)}=\int\bu\,\Id\boldsymbol{\theta}\Id x$, see above,~--- generate
the $N{=}2$ supersymmetric equation
\begin{equation}\label{Burg}
\bu_\xi=\cD_1\cD_2\bu_x+4\bu\bu_x=
  \hat{\boldsymbol{P}}_3\left(\frac{\delta}{\delta\bu}(\boldsymbol{\cH}^{(0)})\right) 
= \hat{\boldsymbol{P}}_2\left(\frac{\delta}{\delta\bu}(\boldsymbol{\cH}^{(1)})\right) 
= \hat{\boldsymbol{P}}_1\left(\frac{\delta}{\delta\bu}(\boldsymbol{\cH}^{(2)})\right),\quad \xi\equiv t_2.
\end{equation}
Super\/-\/equation~\eqref{Burg} was referred to as the $N{=}2$ `Burgers' equation
in~\cite{N=2Hirota,Kiev2005} due to the recovery of 
$\bu_\xi=\bu_{xx}+4\bu\bu_x$ on the diagonal~$\theta_1=\theta_2$. On the other hand, the bosonic limit of~\eqref{Burg} is the tri\/-\/Hamiltonian `minus' Kaup\/--\/Boussinesq system (see~\cite{Kaup75} or 
\cite{PamukKale,KuperIrish,NutkuPavlov} and references therein)
\begin{equation}\label{BLimBurg}
u_{0;\xi}=\bigl(-u_{12}+2u_0^2\bigr)_x,\qquad
u_{12;\xi}=\bigl(u_{0;xx}+4u_0u_{12}\bigr)_x.
\end{equation}
System~\eqref{BLimBurg} is equivalent to the Kaup\/--\/Broer equation via an invertible substitution.
In these terms, super\/-\/equation~\eqref{Burg} is a 
super\/-\/extension of the Kaup\/--\/Boussinesq system~\cite{BrunelliDas94,KuperSuperLongWaves,PalitChowdhury96}. 
In their turn, 
the first three Poisson structures for~\eqref{SKdV} with $a{=}4$ are 
reduced under~\eqref{BRed} to the respective \emph{local} structures for~\eqref{BLimBurg},
see Proposition~\textup{\ref{ThTriHam}} on p.~\textup{\pageref{ThTriHam}}.
\end{example}

Our interest in the recursive production of the integrals of motion
for~\eqref{SKdV} grew after the discovery, 
see~\cite{N=2Hirota}, of new $n$-\/soliton solutions,
\begin{equation}\label{HirotaForm}
\bu=\mathsf{A}(a)\cdot \cD_1\cD_2\log\Bigl(1+\sum_{i=1}^n \alpha_i\exp
 \bigl(k_ix-k_i^3\cdot t\pm\rmi\,k_i\cdot\theta_1\theta_2\bigr)\Bigr),
\qquad
\mathsf{A}(a)=\left\{\begin{matrix}1,\ a{=}1,\\ \tfrac{1}{2},\ a{=}4,
\end{matrix}\right.
\end{equation}
for the 
super\/-\/equations~\eqref{SKdV} with $a{=}1$ or $a{=}4$ (but not $a{=}-2$
or any other $a\in\BBR\setminus\{1,4\}$).
In formula~\eqref{HirotaForm}, the wave numbers~$k_i\in\BBR$ are
arbitrary, and the phases~$\alpha_i$ can be rescaled to $+1$ for
non\/-\/singular $n$-\/soliton solutions 
by appropriate shifts of $n$~higher
times in the SKdV hierarchy. A spontaneous decay of fast solitons and
their transition into the virtual states, on the emerging background of
previously invisible, slow solitons, look paradoxal for such 
KdV\/-\/type systems ($a{=}1$ or $a{=}4$), since they possess an infinity of 
the integrals of motion.

The new solutions~\eqref{HirotaForm} of~\eqref{SKdV} with $a{=}1$ or $a{=}4$
are subject to the condition~\eqref{BRed}
and therefore
satisfy the bosonic limits of these $N{=}2$ super\/-\/systems.
In the same way, 
the bosonic limit~\eqref{BLimBurg} of~\eqref{Burg} admits 
multi\/-\/soliton solutions in Hirota's form~\eqref{HirotaForm}, now with the
exponents $\eta_i=k_ix\pm\rmi k_i^2\xi\pm\rmi k_i\theta_1\theta_2$, 
see~\cite{N=2Hirota}.
This makes the role of such two\/-\/component bosonic reductions particularly important. 
We recall that the reduction~\eqref{BRed} of~\eqref{SKdV} with~$a{=}1$
yields the Kersten\/--\/Kra\-sil'\-shchik equation,
see~\cite{JKKerstenEq} or~\cite{N=2Hirota} and references therein.
In this paper, we consider the bosonic limit of the $N{=}2$, $a{=}4$~SKdV 
equation,
\begin{subequations}\label{BLim}
\begin{align}
u_{0;t}&=-u_{0;xxx}+12 u_0^2 u_{0;x}
   -6\bigl(u_0u_{12}\bigr)_x,\\
u_{12;t}&=-u_{12;xxx}-6u_{12}u_{12;x}
 +12 u_{0;x}u_{0;xx}+6u_0u_{0;xxx}
 +12\bigl(u_0^2u_{12}\bigr)_x,
\end{align}
\end{subequations}
which succeeds the Kaup\/--\/Boussinesq equation~\eqref{BLimBurg} 
in its tri\/-\/Hamiltonian hierarchy. 
We construct a new Gardner deformation for~it (c.f.~\cite{PamukKale}).

In general, system~\eqref{SKdVComponents} with~$a{=}4$ admits three one\/-\/component reductions (except $u_0\not\equiv0$) and three two\/-\/component reductions, which are indicated by the edges that connect the remaining components in the diagram
\[
\begin{CD}
{} @. u_0 @. {} \\
@.    @|     @. \\
u_1 @= u_{12} @= u_2.
\end{CD}
\]
System~\eqref{SKdVComponents} with~$a{=}4$ has no three\/-\/component reductions obtained by setting to zero only one of the four fields 
in~\eqref{N2SuperField}.
We conclude this paper by presenting 
a Gardner deformation for the two\/-\/component boson\/-\/fermion 
reduction~$u_0\equiv0$, $u_2\equiv0$   
of the $N{=}2$,\ $a{=}4$--\/SKdV system, see~\eqref{DefFB} on p.~\pageref{DefFB}.

\section{Deformation problem for $N{=}2$, $a{=}4$--\/SKdV equation}\label{SecGardner}
\noindent%
In this section, we formulate the two\/-\/step 
algorithm for a recursive production of 
the bosonic super\/-\/Hamiltonians~$\boldsymbol{\cH}^{(k)}[\bu]$
for the $N{=}2$ supersymmetric $a{=}4$--\/SKdV hierarchy.
Essentially, we convert the geometric problem to an explicit computational procedure. Our scheme can be applied to other KdV\/-\/type super\/-\/systems (in particular, to~\eqref{SKdV} with $a{=}-2$ or~$a{=}1$).

By definition, a classical Gardner's deformation
for an integrable evolutionary equation~$\cE$ is the diagram
\[\gm_\epsilon\colon\cE(\epsilon)\to\cE,\] where the equation~$\cE(\epsilon)$
is a parametric extension of the initial system $\cE=\cE(0)$ and
$\gm_\epsilon$ is the Miura contraction~\cite{Gardner,KuperIrish,TMPh2006}.
Under the assumption that $\cE(\epsilon)$ be in the form
of a (super-)\/conserved current, the Taylor
coefficients $\tilde{\bu}^{(k)}$ of the formal power series
$\tilde{\bu}=\sum_{k=0}^{+\infty}\tilde{\bu}^{(k)}\cdot\epsilon^k$ are
termwise conserved on~$\cE(\epsilon)$ and hence on~$\cE$. Therefore,
the contraction~$\gm_\epsilon$ yields the recurrence relations, ordered
by the powers of~$\epsilon$, between these densities~$\tilde{\bu}^{(k)}$,
while the equality $\cE(0)=\cE$ specifies its initial condition.

\begin{example}[\textup{\textmd{\cite{Gardner}}}]\label{ExKdVe}
The contraction
\begin{subequations}\label{DefKdV}
\begin{align}
\gm_\epsilon&=\smash{\bigl\{u_{12}=\tu_{12}
   \pm\epsilon\tu_{12;x}-\epsilon^2\tu_{12}^2\bigr\}}\label{KdVeKdV}\\
\intertext{maps solutions $\tu_{12}(x,t;\epsilon)$ of
the extended equation~$\cE(\epsilon)$,}
\tu_{12;t}&+
\bigl(\tu_{12;xx}+3\tu_{12}^2
   -2\epsilon^2\cdot\tu_{12}^3\bigr)_x=0,\label{KdVe}
\end{align}
\end{subequations}
to solutions $u_{12}(x,t)$ of the Korteweg\/--\/de Vries equation~
\eqref{KdV}.
Plugging the 
series $\tu_{12}=\sum_{k=0}^{+\infty}u_{12}^{(k)}\cdot\epsilon^k$
in~$\gm_\epsilon$ for~$\tu_{12}$, we obtain the chain of equations ordered by the powers of~$\epsilon$,
\[
u_{12}=\sum_{k=0}^{+\infty}\tu_{12}^{(k)}\cdot\epsilon^k 
  \pm \tu_{12;x}^{(k)}\cdot\epsilon^{k+1}
   -\sum_{\substack{i+j=k\\ i,j\geq0}} \tu_{12}^{(i)}\tu_{12}^{(j)}\cdot
    \epsilon^{k+2}.
\]
Let us fix the plus sign in~\eqref{KdVeKdV} by reversing $\epsilon\to-\epsilon$ if necessary.
Equating the co\-ef\-fi\-ci\-ents
of
~$\epsilon^k$, we obtain the relations
\[
u=\tu_{12}^{(0)},\qquad
0=\tu_{12}^{(1)}+\tu_{12;x}^{(0)},\qquad
0=\tu_{12}^{(k)}+\tu_{12;x}^{(k-1)}-\sum_{\substack{i+j=k-2\\ i,j\geq0}}
   \tu_{12}^{(i)}\tu_{12}^{(j)},\quad k\geq2.
\]
Hence, from the initial condition $\tu_{12}^{(0)}= u_{12}$,
we recursively generate the densities
\begin{align*}
\tu_{12}^{(1)} &= - u_{12;x},\qquad
\tu_{12}^{(2)} = u_{12;xx} - u_{12}^2,\qquad 
\tu_{12}^{(3)} = - u_{12;xxx} + 4u_{12;x}u_{12},\\
\tu_{12}^{(4)} &= u_{12;4x} - 6u_{12;xx}u_{12} - 5u_{12;x}^2 + 2u_{12}^3,\\
\tu_{12}^{(5)} &= - u_{12;5x} + 8u_{12;xxx}u_{12} + 18u_{12;xx}u_{12;x} - 16u_{12;x}u_{12}^2,\\
\tu_{12}^{(6)} &= u_{12;6x} - 10u_{12;4x}u_{12} - 28u_{12;xxx}u_{12;x} - 19u_{12;xx}^2 + 30u_{12;xx}u_{12}^2 + 50u_{12;x}^2u_{12} - 5u_{12}^4,\\
\tu_{12}^{(7)} &= - u_{12;7x} + 12u_{12;5x}u_{12} + 40u_{12;4x}u_{12;x} + 68u_{12;xxx}u_{12;xx} - 48u_{12;xxx}u_{12}^2\\
 {}&{}\qquad{} - 216u_{12;xx}u_{12;x}u_{12} - 60u_{12;x}^3 + 64u_{12;x}u_{12}^3,\quad \text{\textit{etc}}.
\end{align*}
The conservation $\tu_{12;t}=\tfrac{\Id}{\Id x}\bigl(\cdot\bigr)$ implies
that each coefficient~$u_{12}^{(k)}$ is conserved on~\eqref{KdV}.

The densities $u_{12}^{(2k)}=c(k)\cdot u_{12}^k+\dots$, $c(k)=\text{const}$,
determine the Hamiltonians~$\cH_{12}^{(k)}=\int h_{12}^{(k)}[u_{12}]\,\Id x$
of the renowned KdV hierarchy.
Let us show that all of them are nontrivial.
Consider the zero\/-\/order part~$\breve{u}_{12}^{\text{KdV}}$ such that 
$\tu_{12}\bigl([u_{12}],\epsilon\bigr)=\breve{u}_{12}^{\text{KdV}}(u_{12},\epsilon)+\dots$,
where the dots denote summands containing derivatives of~$u_{12}$. 
Taking the zero\/-\/order component of~\eqref{KdVeKdV},
we conclude that the generating function~$\breve{u}_{12}^{\text{KdV}}$ satisfies the algebraic recurrence relation $u_{12}=\breve{u}_{12}^{\text{KdV}}-\epsilon^2\bigl(\breve{u}_{12}^{\text{KdV}}\bigr)^2$. 
We choose the root by the initial condition
$\breve{u}_{12}^{\text{KdV}}{\bigr|}_{\epsilon=0}=u_{12}$, which yields
\begin{equation}\label{GenFnKdV}
\breve{u}_{12}^{\text{KdV}}=\Bigl( 1-\sqrt{1-4\epsilon^2 u_{12}}
   \bigr)\bigr/(2\epsilon^2).
\end{equation}
Moreover, 
the Taylor coefficients $\breve{u}_{12}^{(k)}(u_{12})$
in $\breve{u}_{12}^{\text{KdV}}=\sum_{k=0}^{+\infty}\breve{u}_{12}^{(k)}\cdot
\epsilon^{2k}$ equal $c(k)\cdot u_{12}^{k+1}$, where $c(k)
$ are positive and grow with~$k$. This is readily seen by induction over~$k$ with the base $\breve{u}_{12}^{(0)}=u_{12}$. Expanding both sides of the equality $u_{12}=\breve{u}_{12}^{\text{KdV}}-\epsilon^2\cdot\bigl(\breve{u}_{12}^{\text{KdV}}\bigr)^2$ in~$\epsilon^2$, we notice that
\[
\breve{u}_{12}^{(k)}=\sum_{\substack{i+j=k-1,\\ i,j\geq0}}
   \breve{u}_{12}^{(i)}\cdot\breve{u}_{12}^{(j)}
 =\sum_{i+j=k-1} c(i)c(j)\cdot u_{12}^{k+1}.
\]
Therefore, the next coefficient, $c(k)=\sum_{i+j=k-1}c(i)\cdot c(j)$,
is the sum over $i,j\geq0$ of products of positive numbers,
whence $c(k+1)>c(k)>0$. This proves the claim.

Let us list the densities 
$h_{\text{KdV}}^{(k)}\sim u_{12}^{(2k)}\mod\img\Id/\Id x$
of the first seven Hamiltonians for~\eqref{KdV}. These will be correlated 
in section~\ref{SecHam} with the lowest seven Hamiltonians for~\eqref{SKdV},
see~\cite{MathieuNew} and~\eqref{Hamsa} below. We have
\begin{align*} 
h_{\text{KdV}}^{(1)} &=   u_{12}^2, \qquad
h_{\text{KdV}}^{(2)} =  2u_{12}^3 - u_{12;x}^2 + 
   2u_{12}^3 + u_{12;xx}, \qquad
h_{\text{KdV}}^{(3)} =  5u_{12}^4 +  5u_{12;xx}u_{12}^2 + u_{12;xx}^2,  \\
h_{\text{KdV}}^{(4)} &=  - 14u_{12}^5 + 70u_{12}^2u_{12;x}^2 + 14u_{12}u_{12;xxx}u_{12;x} + u_{12;xxx}^2,  \\
h_{\text{KdV}}^{(5)} &=  42u_{12}^6 - 420u_{12}^3u_{12;x}^2 + 9u_{12}^2u_{12;6x} + 126u_{12}^2u_{12;xx}^2 + u_{12;4x}^2 - 7u_{12;xx}^3 - 35u_{12;x}^4, \\
h_{\text{KdV}}^{(6)} &=  1056u_{12}^7 - 18480u_{12}^4u_{12;x}^2 + 7392u_{12}^3u_{12;xx}^2 + 55u_{12}^2u_{12;8x} - 1584u_{12}^2u_{12;xxx}^2\\
{}&{}\qquad{} + 66u_{12}u_{12;4x}^2 + 3520u_{12}u_{12;xx}^3 
 - 6160u_{12}u_{12;x}^4 - 8u_{12;5x}^2 + 3696u_{12;xx}^2u_{12;x}^2, \\
h_{\text{KdV}}^{(7)} &=  15444u_{12}^8 - 432432u_{12}^5u_{12;x}^2 + 4004u_{12}^4u_{12;6x} + 216216u_{12}^4u_{12;xx}^2 + 2145u_{12}^3u_{12;8x} \\
{}&{}\qquad - 45760u_{12}^3u_{12;xxx}^2 
 + 3861u_{12}^2u_{12;4x}^2 + 133848u_{12}^2u_{12;xx}^3 - 360360u_{12}^2u_{12;x}^4 \\
{}&{}\qquad{}- 936u_{12}u_{12;5x}^2 + 36u_{12;6x}^2 + 6552u_{12;4x}^2u_{12;xx} 
 + 72072u_{12;xxx}^2u_{12;x}^2 - 28314u_{12;xx}^4.
\end{align*}
At the same time, the densities
$u_{12}^{(2k+1)}=\tfrac{d}{dx}\bigl(\cdot\bigr)\sim0$ are trivial. Indeed, for
$\omega_0\mathrel{{:}{=}}\sum_{k=0}^{+\infty}u_{12}^{(2k)}\cdot\epsilon^{2k}$
and
$\omega_1\mathrel{{:}{=}}\sum_{k=0}^{+\infty}u_{12}^{(2k+1)}\cdot\epsilon^{2k}$
such that $\tu=\omega_0+\epsilon\cdot\omega_1$,
we equate the odd powers of~$\epsilon$ in~\eqref{KdVeKdV} and
obtain $\omega_1=\tfrac{1}{2\epsilon^2}\tfrac{d}{dx}
\log\bigl(1-2\epsilon^2\omega_0\bigr)$.

In what follows, using the deformation
~\eqref{DefKdV} of~\eqref{KdV}, we fix the 
coefficients of differential monomials in~$u_{12}$ within a bigger 
deformation problem (see section~\ref{SecBurg}) for the two\/-\/component
system~\eqref{BLim}.
\end{example}

We split the Gardner deformation problem for the $N{=}2$ supersymmetric
hierarchy of~\eqref{SKdV} with $a{=}4$ 
in two main and several auxiliary 
steps.

First, we note that Miura's contraction
$\gm_\epsilon\colon\cE(\epsilon)\to\cE$, which encodes the recurrence
relation between the conserved densities, is common for all
equations 
of the hierarchy. Indeed, the densities (and hence any differential
relations between them) are shared by all the equations. Therefore, we
pass to the deformation problem for the $N{=}2$ super\/-\/Burgers
equation~\eqref{Burg}. This makes the first simplification of the 
Gardner deformation problem for the $N{=}2$,\ $a{=}4$ super\/-\/KdV hierarchy.

Second, let $\bh^{(k)}$ be an $N{=}2$ super\/-\/conserved density for an evolutionary super\/-\/equation~$\cE$, meaning that 
its velocity w.r.t.\ a time~$\tau$, $\tfrac{\Id}{\Id\tau}{\bh}^{(k)}=
\cD_1(\dots)+\cD_2(\dots)$, is a total divergence on~$\cE$. 
By definition of~$\cD_i$, see~\eqref{SKdV},
the $\theta_1\theta_2$-\/component $h_{12}^{(k)}$ of such
$\bh^{(k)}=h_0^{(k)}+\theta_1\cdot h_1^{(k)}+\theta_2\cdot h_2^{(k)}+
\theta_1\theta_2\cdot h_{12}^{(k)}$ is conserved in the classical 
sense, $\tfrac{\Id}{\Id\tau}{h}_{12}^{(k)}=
\tfrac{\Id}{\Id x}(\dots)$ on~$\cE$. Let us consider the correlation between the conservation laws for the full $N{=}2$ super\/-\/system~$\cE$ and for its reductions that are obtained by setting certain component(s) of~$\bu$ to zero.
In what follows, we study the bosonic reduction~\eqref{BRed}.
Other reductions of the super\/-\/equation~\eqref{SKdV} are discussed in section~\ref{SecHam}, see~\eqref{FB} on p.~\pageref{FB}.

We suppose that the bosonic limit $\lim_B\cE$ of the super\/-\/equation~$\cE$ exists, which is the case for~\eqref{SKdV} and~\eqref{Burg}.
By the above, each conserved super\/-\/density $\bh^{(k)}[\bu]$ 
determines the conserved density $h_{12}^{(k)}[u_0,u_{12}]$, which may become trivial. As in~\cite{BonoraKrivonosSorin}, we assume that the super\/-\/system~$\cE$ does not admit any 
conserved super\/-\/densities that vanish under the reduction~\eqref{BRed}.
Then, for such $h_{12}^{(k)}$ that originates from $\bh^{(k)}$ by construction,
the equivalence class $\{\bh^{(k)}\mod\img\cD_i\}$ is uniquely determined by 
\[
\int h_{12}^{(k)}[u_0,u_{12}]\,\Id x=\int\bh^{(k)}[\bu]{\bigr|}_{u_1=u_2=0}\Id\boldsymbol{\theta}\Id x,\qquad\text{here $N{=}2$ and
$\Id\boldsymbol{\theta}=\Id\theta_1\Id\theta_2$.}
\]
Berezin's 
definition of a super\/-\/integration, $\int\Id\theta_i=0$ and $\int\theta_i\,\Id\theta_i=1$, implies that the problem of recursive generation of the $N{=}2$ super\/-\/Hamiltonians $\boldsymbol{\cH}^{(k)}=\int\bh^{(k)}\,\Id\boldsymbol{\theta}\Id x$ for the SKdV hierarchy amounts to the generation of the equivalence classes $\int h_{12}^{(k)}\,\Id x$ for the respective $\theta_1\theta_2$-\/component. We conclude that a solution of Gardner's deformation problem for the supersymmetric system~\eqref{Burg} may not be subject to the supersymmetry invariance. This is a key point to further reasonings.

We stress that the equivalence class 
of such functions~$h_{12}^{(k)}[u_0,u_{12}]$ that originate 
from~$\boldsymbol{\cH}^{(k)}$ by~\eqref{BRed} is, generally, much more narrow than the equivalence class $\{h_{12}^{(k)}\mod\img \Id/\Id x\}$ of all conserved densities for the bosonic limit~$\lim_B\cE$. Obviously, there are differential functions of the form 
$\tfrac{\Id}{\Id x}\bigl(f[u_0,u_{12}]\bigr)$ 
that can not be obtained\footnote{Under the assumption of weight homogeneity, the freedom in the choice of such $f[u_0,u_{12}]$ is descreased,
but the gap still remains.}
as the $\theta_1\theta_2$-\/component of any
$\bigl[\cD_1(\cdot)+\cD_2(\cdot)\bigr]{\bigr|}_{u_1=u_2=0}$, which is trivial in the super\/-\/sense. Therefore, let $h_{12}^{(k)}$ be \emph{any} recursively given sequence of integrals of motion for~$\lim_B\cE$ 
(\textit{e.g.}, suppose that they are the densities of the Hamiltonians~$\cH^{(k)}$ for the hierarchy of~$\lim_B\cE$), and let it be known that 
each~$\cH^{(k)}=\int h_{12}^{(k)}\,\Id x$ does correspond to the super\/-\/analogue $\boldsymbol{\cH}^{(k)}=\int\bh^{(k)}\,\Id\boldsymbol{\theta}\Id x$.
Then the reconstruction of~$\bh^{(k)}$ requires an intermediate step, which is the elimination of excessive, 
homologically trivial terms under $\Id/\Id x$ that preclude a given $h_{12}^{(k)}$ to be extended to the full super\/-\/density in terms of the $N{=}2$ super\/-\/field~$\bu$. This is illustrated in section~\ref{SecHam}.

Thirdly, the gap between the two types of equivalence for the integrals of motion 
manifests the distinction between the deformations 
$\bigl(\lim_B\cE\bigr)(\epsilon)$ of bosonic limits and, on the other hand,
the bosonic limits $\lim_B\cE(\epsilon)$ of $N{=}2$ 
super\/-\/deformations.
The two operations, Gardner's extension of~$\cE$ to~$\cE(\epsilon)$ and taking the bosonic limit $\lim_B\cF$ of an equation~$\cF$, are not permutable. The resulting systems can be different.
Namely, according to the classical scheme (\cite{Gardner}, 
\cite{TMPh2006}), \emph{each} equation in the evolutionary system
$\bigl(\lim_B\cE\bigr)(\epsilon)$ represents a conserved current, whence each Taylor coefficient of the respective field is conserved, 
see Example~\ref{ExKdVe}. 
At the same time, for $\lim_B\cE(\epsilon)$, the conservation is required only for the field~$\tu_{12}(\epsilon)$, which is the $\theta_1\theta_2$-\/component of the extended super\/-\/field~$\tilde{\bu}(\epsilon)$. Other equations in $\lim_B\cE(\epsilon)$ can have any form.\footnote{Still, the four
components of the
original $N{=}2$ supersymmetric equations within the hierarchy of~\eqref{SKdV}
\emph{are} written in the form of conserved currents. 
A helpful counter\/-\/example, Gardner's extension of 
the $N=1$ super\/-\/KdV equation, is discussed
in~\cite{MathieuNew,MathieuN=1}.}

In this notation, we strengthen the problem of recursive generation of the super\/-\/Hamiltonians for the $N{=}2$ super\/-\/equation~\eqref{Burg}.
Namely, in section~\ref{SecBurg} we construct true
Gardner's deformations for its two\/-\/component bosonic 
limit~\eqref{BLimBurg}.
Moreover, the known deformation~\eqref{DefKdV}
for~\eqref{KdV} upon the component~$u_{12}$ of~\eqref{N2SuperField}
allows to fix 
the coefficients of the terms 
that contain only~$u_{12}$ or its derivatives.
The solution to the Gardner deformation problem generates the recurrence
relation between the nontrivial conserved densities~$h_{12}^{(k)}$
which, in the meantime, depend on~$u_0$ and~$u_{12}$.
By correlating them with the $\theta_1\theta_2$-\/components of the
super\/-\/densities~$\bh^{(k)}$ that depend on~$\bu$, 
we derive the Hamiltonians~$\boldsymbol{\cH}^{(k)}$, $k\geq0$,
for the $N{=}2$ supersymmetric $a{=}4$--\/KdV hierarchy,
see section~\ref{SecHam}.

\section{New deformation of 
the Kaup\/--\/Boussinesq equation}\label{SecBurg}
\noindent%
In this section, we construct a new Gardner's deformation
$\gm_\epsilon\colon\bigl(\lim_B\cE\bigr)(\epsilon)\to\lim_B\cE$
for the `minus' Kaup\/--\/Boussinesq equation~\eqref{BLimBurg}, which is the
bosonic limit of the $N{=}2$ supersymmetric system~\eqref{Burg}.
We will use the known deformation~\eqref{DefKdV} to fix
several coefficients in the Miura contraction~$\gm_\epsilon$,
which ensures the difference of the new solution~\eqref{BLimMiura}--\eqref{BLimBurgE} from previously known deformations of~\eqref{BLimBurg},
see~\cite{PamukKale}.
We prove that the new deformation is maximally nontrivial: It yields infinitely many nontrivial conserved densities, and none of the Hamiltonians is lost.

In components, the $N{=}2$ super\/-\/equation~\eqref{Burg} reads
\begin{align*}
u_{0;\xi}&=\bigl(-u_{12}+2u_0^2\bigr)_x, &
u_{1;\xi}&=\bigl(u_{2;x}+4u_0u_1\bigr)_x,\\
u_{2;\xi}&=\bigl(-u_{1,x}+4u_0u_2\bigr)_x, &
u_{12;\xi} &= \bigl(u_{0;xx}+4u_0u_{12}-4u_1u_2\bigr)_x.
\end{align*}
Clearly, it admits the reduction~\eqref{BRed}; moreover, 
the Kaup\/--\/Boussinesq system~\eqref{BLimBurg}  
is the only possible limit for~\eqref{Burg}. 
Let us summarize its well\/-\/known properties~\cite{Kaup75,NutkuPavlov}:

\begin{state}\label{ThTriHam}
The completely integrable Kaup\/--\/Boussinesq 
system~\eqref{BLimBurg} 
inherits the local tri\/-\/Hamiltonian structure from the 
the two local ($\hat{P}_1$ and $\hat{P}_2$) and the nonlocal
$\hat{P}_3=\hat{P}_2\circ\hat{P}_1\circ\hat{P}_2$ operators for the
$N{=}2$, $a{=}4$--\/SKdV hierarchy under the bosonic limit~\eqref{BRed}\textup{:}
\begin{multline*}
\binom{u_0}{u_{12}}_\xi
 =\hat{A}_1^{12}\binom{\delta/\delta u_0}{\delta/\delta u_{12}}
   \Bigl(\int
   \bigl[2u_0^2u_{12}-\tfrac{1}{2}u_{12}^2-\tfrac{1}{2}u_{0;x}^2\bigr]\,\Id x\Bigr)\\
 =\hat{A}_1^{0}\binom{\delta/\delta u_0}{\delta/\delta u_{12}}
   \Bigl(-\int u_0u_{12}\,\Id x\Bigr)
 =\hat{A}_2\binom{\delta/\delta u_0}{\delta/\delta u_{12}}
    \Bigl(-\int u_{12}\,\Id x\Bigr).
\end{multline*}
The senior Hamiltonian operator~$\hat{A}_2$ is
\[
\begin{pmatrix}
u_{0;x}+2u_0\,\tfrac{\Id}{\Id x} &
 u_{12;x}-4u_0u_{0;x}-2u_0^2\,\tfrac{\Id}{\Id x}+2u_{12}\,\tfrac{\Id}{\Id x}
    +\tfrac{1}{2}\left(\tfrac{\Id}{\Id x}\right)^3\vphantom{\Bigr)}\\
u_{12;x}-2u_0^2\,\tfrac{\Id}{\Id x}+2u_{12}\,\tfrac{\Id}{\Id x}
    +\tfrac{1}{2}\left(\tfrac{\Id}{\Id x}\right)^3 &
 -4u_0u_{12}\,\tfrac{\Id}{\Id x}-4\tfrac{\Id}{\Id x}\circ u_0u_{12}
    -u_0\,\left(\tfrac{\Id}{\Id x}\right)^3 
    -\left(\tfrac{\Id}{\Id x}\right)^3\circ u_0
\end{pmatrix}.
\]
The junior Hamiltonian operators~$\hat{A}_1^{0}$ and~$\hat{A}_1^{12}$
are obtained from~$\hat{A}_2$ by the shifts of the respective 
fields, c.f.~\textup{\cite{Magri2000-5,Artur}}\textup{:}
\begin{align*}
\hat{A}_1^0&=\begin{pmatrix}
\tfrac{\Id}{\Id x} & -2u_{0;x}-2u_0\,\tfrac{\Id}{\Id x} \\
-2u_0\,\tfrac{\Id}{\Id x} & 
  -2u_{12;x}-4u_{12}\,\tfrac{\Id}{\Id x}-\left(\tfrac{\Id}{\Id x}\right)^3
\end{pmatrix}
 =\frac{1}{2}\cdot{\left.\frac{\Id}{\Id\lambda}\right|}_{\lambda=0}
    \hat{A}_2{\Bigr|}_{u_0+\lambda}\\
\intertext{and}
\hat{A}_1^{12}&=\begin{pmatrix} 0 & \tfrac{\Id}{\Id x} \\
   \tfrac{\Id}{\Id x} & 0\end{pmatrix}
 =\frac{1}{2}\cdot{\left.\frac{\Id}{\Id\mu}\right|}_{\mu=0}
    \hat{A}_2{\Bigr|}_{u_{12}+\mu}.
\end{align*}
The three operators $\hat{A}_1^0$, $\hat{A}_1^{12}$, and~$\hat{A}_2$ are Poisson compatible.
\end{state}

The Kaup\/--\/Boussinesq equation~\eqref{BLimBurg} admits an infinite sequence of integrals of motion. We will derive them via the Gardner deformation. Unlike in~\cite{PamukKale}, from now on we always assume that~\eqref{KdVeKdV} is recovered under~$\tu_0\equiv0$.

We assume that both the extension~$\cE(\epsilon)$ of~\eqref{BLimBurg} and the contraction~$\gm_\epsilon\colon\cE(\epsilon)\to\cE$ into~\eqref{BLimBurg} are homogeneous polynomials in~$\epsilon$. From now on, we denote 
the reduction~\eqref{BLimBurg} by~$\cE$.

First, let us estimate the degrees in~$\epsilon$ for such polynomials $\cE(\epsilon)$ and~$\gm_\epsilon$, by balancing the powers of~$\epsilon$ in the left-{} and right\/-\/hand sides of~\eqref{BLimBurg} with $u_0$ and~$u_{12}$
replaced by the Miura contraction 
$\gm_\epsilon=\bigl\{u_0=u_0\bigl([\tu_0,\tu_{12}],\epsilon\bigr)$,
$u_{12}=u_{12}\bigl([\tu_0,\tu_{12}],\epsilon\bigr)\bigr\}$. 
The time evolution in the left\/-\/hand side, which is of the form 
$u_\xi=\partial_{\tu_\xi}(\gm_\epsilon)$ by the chain rule, 
sums the degrees in~$\epsilon$:
$\deg u_\xi=\deg\gm_\epsilon+\deg\cE(\epsilon)$. At the same time,
we notice that system~\eqref{BLimBurg} is only quadratic\/-\/nonlinear.
Hence its right\/-\/hand side, with $\gm_\epsilon$ substituted for $u_0$ 
and~$u_{12}$, gives the degree $2\times\deg\gm_\epsilon$, 
irrespective of~$\deg\cE(\epsilon)$. 
Consequently, we obtain the balance\footnote{This estimate is rough and can be
improved by operating separately with the components of~$\gm_\epsilon$ 
and~$\cE(\epsilon)$ since, in particular, the Kaup\/--\/Boussinesq 
system~\eqref{BLimBurg} 
is \emph{linear} in~$u_{12}$.} $1:1$ for $\max\deg\gm_\epsilon:\max\deg\cE(\epsilon)$. 
This is in contrast with the balance
$1:2$ for polynomial deformations of the bosonic limit~\eqref{BLim} for the initial SKdV system~\eqref{SKdV}, which is cubic\/-\/nonlinear\footnote{Reductions other than~\eqref{BRed} can produce quadratic\/-\/nonlinear subsystems of the cubic\/-\/nonlinear system~\eqref{SKdV}, 
\textit{e.g.}, if one sets $u_0=0$ and $u_2=0$, see~\eqref{FB} on p.~\pageref{FB}.}
(c.f.\ \cite{MathieuNew}).

Obviously, a lower degree polynomial extension~$\cE(\epsilon)$ contains fewer undetermined coefficients. This is the first profit we gain from passing to~\eqref{Burg} instead of~\eqref{SKdV}. By the same argument, we conclude that
$\gm_\epsilon\colon\cE(\epsilon)\to\cE$, viewed as the algebraic system upon these coefficients, is only \emph{quadratic}\/-\/nonlinear w.r.t.\ the coefficients in~$\gm_\epsilon$ (and, obviously, \emph{linear} w.r.t.\ the coefficients in~$\cE(\epsilon)$; this is valid for any balance $\deg\gm_\epsilon:\deg\cE(\epsilon)$). Hence the size of this overdetermined algebraic system is further decreased.

Second, we use the unique admissible homogeneity weights for the Kaup\/--\/Boussinesq system~\eqref{BLimBurg},
\[
|u_0|=1,\quad |u_{12}|=2,\quad |\Id/\Id\xi|=2;
\]
here $|\Id/\Id x|\equiv 1$ is the normalization.
The Miura contraction
$\gm_\epsilon=\bigl\{u_0=\tu_0+\epsilon\cdot(\dots)$, 
 $u_{12}=\tu_{12}+\epsilon\cdot(\dots)\bigr\}$,
which we assume regular at the origin,
implies that $|\tu_0|=1$ and $|\tu_{12}|=2$ as well. 
We let $|\epsilon|=-1$ be the difference of weights for every two 
successive 
Hamiltonians for the $N{=}2$, $a{=}4$--\/SKdV hierarchy, 
see~\cite{MathieuNew} and~\eqref{Hamsa} below.
In this setup, all functional coefficients of the powers~$\epsilon^k$ both in~$\cE(\epsilon)$ and~$\gm_\epsilon$ are homogeneous differential polynomials in $u_0$, $u_{12}$, and their derivatives w.r.t.\ $x$. It is again important that the time~$\xi$ of weight $|\Id/\Id\xi|=2$ in~\eqref{Burg} precedes the time~$t$
with $|\Id/\Id t|=3$ in the hierarchy of~\eqref{SKdV}, 
where $|\theta_i|=-\tfrac{1}{2}$ and~$|\bu|=1$. 
As before, we have further decreased the number of undetermined coefficients.

The polynomial ansatz for Gardner's deformation of~\eqref{BLimBurg}
is generated by the procedure\footnote{\label{FootGenSSPoly}%
The call is \texttt{GenSSPoly(N,wglist,cname,mode)}, where 
\begin{itemize}
\item \texttt{N} is the number
of Grassmann variables $\theta_1$,\ $\ldots$,\ $\theta_N$;
\item \texttt{wglist} is the list of lists \texttt{\{afwlist, abwlist, wgt\}},
each containing the list \texttt{afwlist} of weights for the fermionic super\/-\/fields
and the list \texttt{abwlist} of weights for the bosonic super\/-\/fields; here 
\texttt{wgt} is the weight of the polynomial to be constructed;
\item \texttt{cname} is the prefix for the names of arising undetermined coefficients (\textit{e.g.}, \texttt{p} produces~$p_1,p_2,\dots$);
\item \texttt{mode} is the list of flags, which can be \texttt{fonly}, 
whence only fermionic polynomials are generated, or \texttt{bonly}, which yields the bosonic output.
\end{itemize}}
\texttt{GenSSPoly}, which is a new possibility 
in the the analytic software~\cite{SsTools}.
We thus obtain the determining system $\gm_\epsilon\colon\cE(\epsilon)\to\cE$.
Using \textsc{SsTools}, we split it to the overdetermined system of algebraic equations, which are linear w.r.t.\ $\cE(\epsilon)$ and quadratic\/-\/nonlinear w.r.t.~$\gm_\epsilon$. Moreover, we claim that this system is 
\emph{triangular}. 
Indeed, it is ordered by the powers of~$\epsilon$, since the determining system is identically satisfied at zeroth order and because equations at lower orders of~$\epsilon$ involve only the coefficients of its lower powers
from~$\gm_\epsilon$ and~$\cE(\epsilon)$.

Thirdly, we use the deformation~\eqref{DefKdV} of the Korteweg\/--\/de Vries equation~\cite{Gardner}. We recall that
\begin{itemize}
\item Miura's contraction $\gm_\epsilon$ is common for all two\/-\/component systems in the bosonic limit, see~\eqref{BRed}, of the $N{=}2$, $a{=}4$--\/SKdV hierarchy;
\item for any $a$, the bosonic limit of~\eqref{SKdV}, 
see~\eqref{SKdVComponents} and~\eqref{BLim},
incorporates the Korteweg\/--\/de Vries equation~\eqref{KdV}.
\end{itemize}
Using~\eqref{KdVeKdV}, we fix those coefficients in~$\gm_\epsilon$ which depend only on~$u_{12}$ and its derivatives, but not on~$u_{0}$ or its derivatives.
Apparently, we discard the knowledge of such coefficients in the extension of the bosonic limit~\eqref{BLim}, 
since for us now it is not the object to be deformed. But the minimization of the algebraic system, which we have achieved by passing 
to~\eqref{Burg}, is so significant that this temporary loss in inessential.
Furthermore, the above reasoning shows that the recovery of the coefficients in the extension~$\cE(\epsilon)$ amounts to solution of linear equations, while finding the coefficients in~$\gm_\epsilon$ would cost us the necessity to solve nonlinear algebraic systems. We managed to fix some of those constants for granted.

We finally remark that the normalization of at least one coefficient in the deformation problem cancels the reduntant dilation of the parameter~$\epsilon$,
which, otherwise, would remain until the end.
This is our fourth simplification.\footnote{There is one more possibility to reduce the size of the algebraic system: this can be achieved by a thorough balance of the \emph{differential orders} 
of~$\gm_\epsilon$ and~$\cE(\epsilon)$.}

We let the degrees $\deg\gm_\epsilon=\deg\cE(\epsilon)$ be equal to four
(c.f.~\cite{MathieuNew}).
Under this assumption, the two\/-\/component homogeneous polynomial extension~$\cE(\epsilon)$ of system~\eqref{BLimBurg} contains~$160$
undetermined coefficients. At the same time, the two components of the Miura
contraction~$\gm_\epsilon$ depend on~$94$
coefficients. However, we decrease this number by nine, setting the coefficient of~$\tu_{12;x}$ equal to ${+}1$ and, similarly, to ${-}1$ for~$\tu_{12}^2$ (see~\eqref{KdVeKdV}, 
where the $\pm$ sign is absorbed by $\epsilon\mapsto {-}\epsilon$).
Likewise, we set equal to zero
the seven coefficients of $\tu_{12;xx}$, $\tu_{12}\tu_{12;x}$, $\tu_{12;xxx}$,
$\tu_{12}^3$, $\tu_{12;x}^2$, $\tu_{12}\tu_{12;xx}$, and~$\tu_{12;xxxx}$
in~$\gm_\epsilon$.

The resulting algebraic system with the shortened list of unknowns and with the auxiliary list of nine substitutions is handled by~\textsc{SsTools}
and then solved by using \textsc{Crack}~\cite{WolfCrack}.

\begin{theor}
Under the above assumptions, the Gardner deformation problem
for the Kaup\/--\/Boussinesq equation~\eqref{BLimBurg}
has   
a unique real solution 
of degree~$4$. The 
Miura contraction~$\gm_\epsilon$ is given by 
\begin{subequations}\label{BLimMiura}
\begin{align}
u_{0}  &= \tu_{0} + \epsilon \tu_{0;x} - 2\epsilon^2\tu_{12}\tu_{0},
   \label{BLimMiurau0}\\
u_{12} &= \tu_{12} + \epsilon\bigl(\tu_{12;x} - 2\tu_{0}\tu_{0;x}\bigr)
  + \epsilon^2\bigl(4\tu_{12}\tu_{0}^2 -\tu_{12}^2  - \tu_{0;x}^2\bigr)
+ 4\epsilon^3 \tu_{12}\tu_{0}\tu_{0;x}
 -4\epsilon^4  \tu_{12}^2\tu_{0}^2.\label{BLimMiurau12}
\end{align}
\end{subequations}
The extension~$\cE(\epsilon)$ of
~\eqref{BLimBurg} 
is
\begin{subequations}\label{BLimBurgE}
\begin{align}
\tu_{0;\xi} &= -\tu_{12;x}+4u_0\tu_{0;x}
 +2\epsilon\bigl(\tu_0\tu_{0;x}\bigr)_x
 -4\epsilon^2\bigl(\tu_0^2u_{12}\bigr)_x,
\\
\tu_{12;\xi}&=\tu_{0;xxx}+4\bigl(\tu_0\tu_{12}\bigr)_x
 -2\epsilon\bigl(\tu_0\tu_{12;x}\bigr)_x
 -4\epsilon^2\bigl(\tu_0\tu_{12}^2\bigr)_x.
\end{align}
\end{subequations}
System~\eqref{BLimBurgE} 
preserves the first Hamiltonian operator
$\hat{A}_1^\epsilon=\left(\begin{smallmatrix} 0 & \Id/\Id x\\ \Id/\Id x & 0\end{smallmatrix}\right)$ from~$\hat{A}_1^{12}$ for~\eqref{BLimBurg}.
\end{theor}

The Miura contraction~$\gm_\epsilon$ 
is shared by all equations in the Kaup\/--\/Boussinesq hierarchy.
Solving the linear algebraic system, we find the
extension $\bigl(\lim_B\cE_{\text{SKdV}}^{a{=}4}\bigr)(\epsilon)$
for the bosonic limit~\eqref{BLim} of~\eqref{SKdV} with $a{=}4$:
\begin{subequations} \label{AppR}
\begin{align}
\tu_{0;t} &=  - \tu_{0;xxx}  - 6\bigl(\tu_{0}\tu_{12}\bigr)_x 
+ 12\tu_{0}^2\tu_{0;x}
  + 12\epsilon \bigl(\tu_{0}^2\tu_{0;x} \bigr)_x 
 + 6\epsilon^2 \bigl( \tu_{0}\tu_{12}^2 - 4\tu_{12}\tu_{0}^3 + 
\tu_{0}\tu_{0;x}^2) \bigr)_x  \notag \\
 {}&{}\qquad{}   + \epsilon^3 \bigl( (-24)\tu_{12}\tu_{0}^2\tu_{0;x}\bigr)_x 
   + \epsilon^4 \bigl( 24\tu_{12}^2\tu_{0}^3 \bigr)_x,\\
\tu_{12;t} &= - \tu_{12;xxx} - 6\tu_{12}\tu_{12;x}
 + 12\bigl(\tu_{0}^2\tu_{12}\bigr)_x + 6\tu_{0}\tu_{0;xxx}  
 + 12\tu_{0;xx}\tu_{0;x} \notag\\
{}&{}\qquad{}
 + 6\epsilon \bigl( \tu_{0;xx}\tu_{0;x} - 2\tu_{0}^2\tu_{12;x}  \bigr)_x \notag\\
{}&{}\qquad{} + 2\epsilon^2 \bigl( \tu_{12}^3 - 18\tu_{12}^2\tu_{0}^2 
   - 6\tu_{12}\tu_{0}\tu_{0;xx} - 3\tu_{12}\tu_{0;x}^2 
   - 6\tu_{0}\tu_{12;x}\tu_{0;x}\bigr)_x\notag\\
{}&{}\qquad{}
 + 24\epsilon^3 \bigl( \tu_{12}\tu_{0}^3\tu_{12;x} \bigr)_x 
 + 24\epsilon^4 \bigl( \tu_{12}^3\tu_{0}^2 \bigr)_x. 
\end{align}
\end{subequations}
Now we expand the fields 
$\tu_0(\epsilon)=\sum_{k=0}^{+\infty}\tu_0^{(k)}\cdot\epsilon^k$ and
$\tu_{12}(\epsilon)=\sum_{k=0}^{+\infty}\tu_{12}^{(k)}\cdot\epsilon^k$, and
plug the formal power series for~$\tu_0$ and~$\tu_{12}$ in~$\gm_\epsilon$.
Hence we start from
$\tu_0^{(0)}=u_0$ and $\tu_{12}^{(0)}=u_{12}$, which is standard, and
proceed with the recurrence relations between the conserved densities
$u_0^{(k)}$ and~$u_{12}^{(k)}$,
\begin{gather*}
\tu_0^{(1)} = - u_{0;x}, \quad 
\tu_0^{(n)} = - \tfrac{\Id}{\Id x}\tu_0^{({n-1})} + \sum\limits_{j+k=n-2}2\tu_{12}^{(k)}\tu_0^{(j)},\quad \forall n\geq2; \\
%
\tu_{12}^{(1)} = 2u_{0}u_{0;x} - u_{12;x},\quad 
\tu_{12}^{(2)} = 
u_{12}^2+u_{12;xx}-4u_{12}u_0^2-3u_{0;x}^2-4u_0u_{0;xx},
\end{gather*}\begin{gather*}
\tu_{12}^{(3)} = \sum\limits_{j+k=2}2\tu_0^{(j)}\tfrac{\Id}{\Id x}\tu_0^{(k)}  
  -\tfrac{\Id}{\Id x}\tu_{12}^{({2})} 
+\sum\limits_{j+k=1}\Bigl(\tu_{12}^{(j)}\tu_{12}^{(k)} + (\tfrac{\Id}{\Id x}\tu_0^{(j)})(\tfrac{\Id}{\Id x}\tu_0^{(k)})\Bigr)\\
- \sum\limits_{j+k+l=1} 4\tu_{12}^{(j)}\tu_0^{(k)}\tu_0^{(l)} - 4u_{12}u_{0}u_{0;x},\\
\tu_{12}^{(n)} = - \tfrac{\Id}{\Id x}\tu_{12}^{({n-1})}
+\sum\limits_{j+k=n-1}2\tu_0^{(j)}\tfrac{\Id}{\Id x}\tu_0^{(k)}  + \sum\limits_{j+k=n-2}(\tu_{12}^{(j)}\tu_{12}^{(k)} + (\tfrac{\Id}{\Id x}(\tu_0^{(j)})
  \tfrac{\Id}{\Id x}(\tu_0^{(k)}))\\ 
-\sum\limits_{j+k+l=n-2} 4\tu_{12}^{(j)}\tu_0^{(k)}\tu_0^{(l)} 
-\sum\limits_{j+k+l=n-3}4\tu_{12}^{(j)}\tu_0^{(k)}\tfrac{\Id}{\Id x}\tu_0^{(l)}\\
+ \sum\limits_{j+k+l+m=n-4}4\tu_{12}^{(j)}\tu_{12}^{(k)}\tu_0^{(l)}\tu_0^{(m)},\qquad \forall n\geq4. 
\end{gather*}

\begin{example}\label{ExBLim7dens}
Following this recurrence, let us generate the eight lowest weight nontrivial conserved densities, which start the tower of Hamiltonians for the 
Kaup\/--\/Boussinesq hierarchy.

We begin with $\tu_0^{(0)}=u_0$ and $\tu_{12}^{(0)}=u_{12}$.
Next, we obtain the densities
\[
\tu_0^{(2)} =u_{0;xx} + 2u_{0}u_{12},\qquad 
\tu_{12}^{(2)}=-4u_{0;xx}u_{0}-3u_{0;x}^2+u_{12;xx}-4u_{0}^2u_{12}+u_{12}^2,
\]
which contribute to
the tri\/-\/Hamiltonian representation of~\eqref{BLimBurg},
see Proposition~\ref{ThTriHam}. Now we proceed with 
\begin{align*}
\tu_0^{(4)}&= u_{0;4x} - 12u_{0;xx}u_{0}^2 + 6u_{0;xx}u_{12} - 18u_{0;x}^2u_{0} + 10u_{0;x}u_{12;x} + 6u_{12;xx}u_{0} - 8u_{0}^3u_{12} + 6u_{0}u_{12}^2, 
\\
\tu_{12}^{(4)}&= - 8u_{0;4x}u_{0} - 20u_{0;xxx}u_{0;x} - 13u_{0;xx}^2 + 32u_{0;xx}u_{0}^3 - 48u_{0;xx}u_{0}u_{12} + 72u_{0;x}^2u_{0}^2 - 38u_{0;x}^2u_{12} -{}\\
{}&{}- 80u_{0;x}u_{12;x}u_{0} + u_{12;4x} - 24u_{12;xx}u_{0}^2 + 6u_{12;xx}u_{12} + 5u_{12;x}^2 + 16u_{0}^4u_{12} - 24u_{0}^2u_{12}^2 + 2u_{12}^3,
\\
\tu_0^{(6)}&= u_{0;6x} - 40u_{0;4x}u_{0}^2 + 10u_{0;4x}u_{12} - 200u_{0;xxx}u_{0;x}u_{0} + 28u_{0;xxx}u_{12;x} - 130u_{0;xx}^2u_{0} - {}  \\
{}&{}- 198u_{0;xx}u_{0;x}^2 + 38u_{0;xx}u_{12;xx} + 80u_{0;xx}u_{0}^4 - 240u_{0;xx}u_{0}^2u_{12} + 30u_{0;xx}u_{12}^2 + 240u_{0;x}^2u_{0}^3 - {} \\
{}&{} - 380u_{0;x}^2u_{0}u_{12} + 28u_{0;x}u_{12;xxx} - 400u_{0;x}u_{12;x}u_{0}^2 + 100u_{0;x}u_{12;x}u_{12} + 10u_{12;4x}u_{0} - {} \\
{}&{}- 80u_{12;xx}u_{0}^3 + 60u_{12;xx}u_{0}u_{12} + 50u_{12;x}^2u_{0} + 32u_{0}^5u_{12} - 80u_{0}^3u_{12}^2 + 20u_{0}u_{12}^3,
\\
\tu_{12}^{(6)}&= - 12u_{0;6x}u_{0} - 42u_{0;5x}u_{0;x} - 80u_{0;4x}u_{0;xx} + 160u_{0;4x}u_{0}^3 - 120u_{0;4x}u_{0}u_{12} - 49u_{0;xxx}^2 + {}\\
{}&{}+ 1200u_{0;xxx}u_{0;x}u_{0}^2 - 312u_{0;xxx}u_{0;x}u_{12} - 336u_{0;xxx}u_{12;x}u_{0} + 780u_{0;xx}^2u_{0}^2 - 206u_{0;xx}^2u_{12} + {}\\
{}&{} + 2376u_{0;xx}u_{0;x}^2u_{0} - 716u_{0;xx}u_{0;x}u_{12;x} - 456u_{0;xx}u_{12;xx}u_{0} - 192u_{0;xx}u_{0}^5 + 960u_{0;xx}u_{0}^3u_{12} -{}\\
{}&{}- 360u_{0;xx}u_{0}u_{12}^2 + 297u_{0;x}^4 - 366u_{0;x}^2u_{12;xx} - 720u_{0;x}^2u_{0}^4 + 2280u_{0;x}^2u_{0}^2u_{12} - 290u_{0;x}^2u_{12}^2 -{}\\
{}&{}- 336u_{0;x}u_{12;xxx}u_{0} + 1600u_{0;x}u_{12;x}u_{0}^3 - 1200u_{0;x}u_{12;x}u_{0}u_{12} + u_{12;6x} - 60u_{12;4x}u_{0}^2 +{}\\
{}&{}+ 10u_{12;4x}u_{12} + 28u_{12;xxx}u_{12;x} + 19u_{12;xx}^2 + 240u_{12;xx}u_{0}^4 - 360u_{12;xx}u_{0}^2u_{12} + 30u_{12;xx}u_{12}^2-{}\\
{}&{} - 300u_{12;x}^2u_{0}^2 + 50u_{12;x}^2u_{12} - 64u_{0}^6u_{12} + 240u_{0}^4u_{12}^2 - 120u_{0}^2u_{12}^3 + 5u_{12}^4,\qquad\text{\textit{etc}}.
\end{align*}
We will use these formulas in the next section, where, as an illustration, we re\/-\/derive the seven super\/-\/Hamiltonians of~\cite{MathieuNew}.
\end{example}
\enlargethispage{0.7\baselineskip}
\begin{theor}
In the above notation, the following statements hold\textup{:}
\begin{itemize}
\item The conserved densities $\tu_0^{(2k)}$ and~$\tu_{12}^{(2k)}$ of weights
$2k+1$ and $2k+2$, respectively, are nontrivial for all integers~$k\geq0$.
\item Consider the zero\/-\/order components $\breve{u}_0(u_0,u_{12},\epsilon)$
and $\breve{u}_{12}(u_0,u_{12},\epsilon)$ of the series 
$\tu_0\bigl([u_0,u_{12}],\epsilon\bigr)$ and
$\tu_{12}\bigl([u_0,u_{12}],\epsilon\bigr)$ with differential\/-\/polynomial coefficients. Then these generating functions are given by the 
formulas
\begin{subequations}
\begin{align}
\left(\breve{u}_0(u_0,u_{12},\epsilon^2)\right)^2&=\frac{1}{8\epsilon^2}\cdot
\left[4\epsilon^2(u_0^2+u_{12})-1+
 \sqrt{1+8\epsilon^2(u_0^2-u_{12})+16\epsilon^4(u_0^2+u_{12})^2}\right],
  \label{u0zero}\\
\breve{u}_{12}(u_0,u_{12},\epsilon^2)&=
\frac{1}{2\epsilon^2}\cdot\left[
 1-\sqrt{\tfrac{1}{2}-2\epsilon^2(u_{12}+u_0^2)
  +\tfrac{1}{2}\sqrt{1+8\epsilon^2(u_0^2-u_{12})+16\epsilon^4(u_0^2+u_{12})^2
}}\right].\label{u12zero}
\end{align}
\end{subequations}
\item The generating functions for the odd\/-\/index
conserved densities $\tu_0^{(2k+1)}$ and~$\tu_{12}^{(2k+1)}$ are expressed
via the even\/-\/index densities, see~\eqref{u0oddTriv} and~\eqref{u12oddTriv},
respectively. We claim that all the odd\/-\/index densities
are trivial. 
\end{itemize}
\end{theor}

\begin{proof}
The densities~$\tu_0^{(k)}$ and~$\tu_{12}^{(k)}$, which are conserved for the bosonic limit~\eqref{BLim} of 
the $N{=}2$,\ $a{=}4$--\/SKdV system~\eqref{SKdVComponents},
retract to the conserved densities for the Korteweg\/--\/de Vries equation~\eqref{KdV} under~$u_0\equiv0$, 
see Example~\ref{ExKdVe}. 
The corresponding reduction of~$\breve{u}_{12}(u_0,u_{12},\epsilon)$ is the
generating function~\eqref{GenFnKdV}. This implies that 
$\breve{u}_{12}=\sum_{k=0}^{+\infty} c(k)u_{12}^k\cdot\epsilon^{2k}+\dots$,
whence the densities~$\tu_{12}^{(2k)}$ are nontrivial.

Following the line of reasonings on p.~\pageref{GenFnKdV}, we consider the
zero\/-\/order terms in Miura's contraction~\eqref{BLimMiura}, which yields
\begin{subequations}
\begin{align}
u_0&=\breve{u}_0\cdot\bigl(1-2\epsilon^2\breve{u}_{12}\bigr),\label{Short0}\\
u_{12}&=\breve{u}_{12}+\epsilon^2\bigl(4\breve{u}_0^2\breve{u}_{12}-\breve{u}_{12}^2\bigr)-4\epsilon^4\breve{u}_0^2\breve{u}_{12}^2.\label{Short12}
\end{align}
\end{subequations}
Therefore,
\[
\breve{u}_0=\frac{u_0}{1-2\epsilon^2\breve{u}_{12}}=
   \sum_{k=0}^{+\infty} u_0\cdot\bigl(2\epsilon^2\breve{u}_{12}\bigr)^k.
\]
Since the coefficients~$c(k)$ of $u_{12}^k\cdot\epsilon^{2k}$ 
in~$\breve{u}_{12}$ are positive, so are the coefficients 
of~$u_0u_{12}^k\cdot\epsilon^{2k}$ in~$\breve{u}_0$ for all~$k\geq0$.
This proves that the conserved densities~$\tu_0^{(2k)}$ are nontrivial as well.

Second, squaring~\eqref{Short0} and adding it to~\eqref{Short12}, we obtain
the equality $u_0^2+u_{12}=\breve{u}_0^2+\breve{u}_{12}-\epsilon^2\breve{u}_{12}^2$. In agreement with $\breve{u}_0{\bigr|}_{\epsilon=0}=u_0$ and
$\breve{u}_{12}{\bigr|}_{\epsilon=0}=u_{12}$, we choose the root
$\breve{u}_{12}=\bigl[1-\sqrt{1-4\epsilon^2\cdot
  \bigl(u_{12}+u_0^2-\breve{u}_0^2\bigr)}\bigr]/(2\epsilon^2)$
of this quadratic equation.
Hence~\eqref{Short0} yields the bi\/-\/quadratic equation upon~$\breve{u}_0$,
\[
1-4\epsilon^2\bigl(u_{12}+u_0^2-\breve{u}_0^2\bigr)=u_0^2\bigr/\breve{u}_0^2.
\]
As above, the proper choice of its root gives~\eqref{u0zero}, whence we return
to~$\breve{u}_{12}$ and finally obtain~\eqref{u12zero}.

Finally, let us substitute the expansions $\tu_0=
\upsilon_0(\epsilon^2)+\epsilon\cdot\upsilon_1(\epsilon^2)$
and $\tu_{12}=\omega_0(\epsilon^2)+\epsilon\cdot\omega_1(\epsilon^2)$
in~\eqref{BLimMiura} for~$\tu_0$ and~$\tu_{12}$, 
see Example~\ref{ExKdVe}. 
By balancing the odd powers of~$\epsilon$ in~\eqref{BLimMiurau0},
it is then easy to deduce the equality
\begin{equation}\label{u0oddTriv}
\upsilon_1\equiv\sum_{k=0}^{+\infty} \tu_0^{(2k+1)}\cdot\epsilon^{2k}=
\frac{1}{4\epsilon^2}\cdot\frac{\Id}{\Id x}\log\bigl(
 1-4\epsilon^2\cdot\upsilon_0\bigr),\qquad
\text{where }
\upsilon_0\equiv\sum_{\ell=0}^{+\infty}\tu_0^{(2\ell)}\cdot\epsilon^{2\ell}.
\end{equation}
The balance of odd powers of~$\epsilon$ in~\eqref{BLimMiurau12}
yields the algebraic equation upon~$\omega_1$, whence, 
in agreement with the initial condition~$\omega_1(0)=\tu_{12}^{(1)}$,
we choose its root
\begin{multline} 
\omega_1=
\Bigl[1-2 \epsilon^2 \omega_0
+4 \epsilon^2 \upsilon_0^2
+4 \epsilon^4 \bigl(\upsilon_1^2
 -2 \omega_0 \upsilon_0^2
 + \upsilon_0 \upsilon_{1;x} 
 + \upsilon_1 \upsilon_{0;x}\bigr)
-8 \epsilon^6 \upsilon_1^2 \omega_0\\
-\Bigl(1
+4 \epsilon^2 \bigl(2\upsilon_0^2
 - \omega_0\bigr)
+4 \epsilon^4 \bigl(\omega_0^2
 +2 \upsilon_0 \upsilon_{1;x}
 -8 \omega_0 \upsilon_0^2
 +2 \upsilon_1 \upsilon_{0;x}
 +2 \upsilon_1^2
 +4 \upsilon_0^4\bigr)\\
+16\epsilon^6\bigl(
2 \omega_0^2 \upsilon_0^2
 -2  \upsilon_1^2 \omega_0
 - \omega_0 \upsilon_0 \upsilon_{1;x}
 - \omega_0 \upsilon_1 \upsilon_{0;x}
 -2 \upsilon_0^2 \upsilon_1 \upsilon_{0;x}
 +2 \upsilon_1 \upsilon_0 \omega_{0;x}
 +2 \upsilon_1^2 \upsilon_0^2
 -4 \omega_0 \upsilon_0^4
 +2 \upsilon_0^3 \upsilon_{1;x}\bigr)\\
+16 \epsilon^8 \bigl(\upsilon_1^4
+2 \omega_0^2 \upsilon_1^2
+4 \omega_0^2 \upsilon_0^4
-2 \upsilon_1^2 \upsilon_0 \upsilon_{1;x}
-4 \omega_0 \upsilon_0^3 \upsilon_{1;x}
+8 \upsilon_1^2 \omega_0 \upsilon_0^2
+2 \upsilon_1^3 \upsilon_{0;x}\\
+ \upsilon_0^2 \upsilon_{1;x}^2
+ \upsilon_1^2 \upsilon_{0;x}^2
+4 \omega_0 \upsilon_0^2 \upsilon_1 \upsilon_{0;x}
-2 \upsilon_0 \upsilon_{1;x} \upsilon_1 \upsilon_{0;x}\bigr)\\
+64 \epsilon^{10} \bigl(\upsilon_0 \upsilon_{1;x} \upsilon_1^2 \omega_0
-2 \omega_0^2 \upsilon_0^2 \upsilon_1^2
- \upsilon_1^3 \upsilon_{0;x} \omega_0
- \upsilon_1^4 \omega_0\bigr)
+64 \epsilon^{12} \upsilon_1^4 \omega_0^2
\Bigr)^{1/2}\Bigr]\bigr/(16 \epsilon^6 \upsilon_1 \upsilon_0).
   \label{u12oddTriv}
\end{multline}
We claim that, using the 
balance of the even powers of~$\epsilon$ in~\eqref{BLimMiura},
the 
representation 
$\sum_{k=0}^{+\infty}\tu_{12}^{(2k+1)}\cdot\epsilon^{2k}\in
\img\tfrac{\Id}{\Id x}$
can be   
deduced, whence $\tu_{12}^{(2k+1)}\sim0$.
\end{proof}

\section{Super\/-\/Hamiltonians for 
$N{=}2$,\ $a{=}4$--\/SKdV hierarchy}\label{SecHam}
\noindent%
In this section, we assign the bosonic super\/-\/Hamiltonians 
$\boldsymbol{\cH}^{(k)}=\int\boldsymbol{h}^{(k)}[\bu]\,\Id\boldsymbol\theta\Id x$
of~\eqref{SKdV} with $a{=}4$ to the Hamiltonians 
$H^{(k)}=\int h_{12}^{(k)}[u_0,u_{12}]\,\Id x$ of its bosonic limit~\eqref{BLim}.
Also, we establish the 
no\/-\/go result on the super\/-\/field, $N{=}2$ supersymmetry invariant deformations of $a{=}4$--\/SKdV
that retract to~\eqref{DefKdV} under the respective reduction in the super\/-\/field~\eqref{N2SuperField}.
At the same time, we initiate the study of Gardner's deformations for
reductions of~\eqref{SKdVComponents} other than~\eqref{BRed}, 
and here we find 
the deformations of two\/-\/component fermion\/-\/boson limit in~it.
However, we observe that the new solutions can not be merged with the deformation~\eqref{AppR} for the bosonic limit of~\eqref{SKdVComponents}.

From the previous section, we know the procedure for recursive production of
the Hamiltonians $H^{(k)}=\int h^{(k)}\,\Id x$ for the bosonic 
limit~\eqref{BLim} of the $N{=}2$, $a{=}4$--\/SKdV equation,
here $h^{(2k)}=\tu_0^{(2k)}$ and~$h^{(2k+1)}=\tu_{12}^{(2k)}$. 
In section~\ref{SecGardner},
we explained why the reconstruction of the densities~$\boldsymbol{h}^{(k)}$
for the bosonic super\/-\/Hamiltonians~$\boldsymbol{\cH}^{(k)}$ 
from $h^{(k)}\bigl[u_0,u_{12}\bigr]$ requires an intermediate step. 
Namely, it amounts to the proper choice of the 
representatives~$h_{12}^{(k)}$ within the equivalence class 
$\bigl\{h^{(k)}\mod\img\tfrac{\Id}{\Id x}\bigr\}$ such that $h_{12}^{(k)}$ 
can be realized under~\eqref{BRed} as the $\theta_1\theta_2$-\/component 
of the super\/-\/density~$\boldsymbol{h}^{(k)}$. This allows to
restore the dependence on the components~$u_1$ and~$u_2$ 
of~\eqref{N2SuperField} and to recover the supersymmetry invariance.
The former means that each $\boldsymbol{h}^{(k)}$~is conserved on~\eqref{SKdVComponents} and the latter implies that $\boldsymbol{h}^{(k)}$~becomes
a differential function in~$\bu$.

The correlation between \emph{unknown} bosonic super\/-\/differential
polynomials $\boldsymbol{h}^{(k)}[\bu]$ and 
the densities~$h^{(k)}\bigl[u_0,u_{12}\bigr]$, which are produced by
the recurrence relation, is established as follows. First, we generate
the homogeneous super\/-\/differential polynomial ansatz for 
the bosonic~$\boldsymbol{h}^{(k)}$ using \texttt{GenSSPoly}, 
see 
note~\ref{FootGenSSPoly} on p.~\pageref{FootGenSSPoly}.
Second, we split 
the super\/-\/field~$\bu$ using the right\/-\/hand side 
of~~\eqref{N2SuperField} and obtain the 
$\theta_1\theta_2$-\/component~$h_{12}^{(k)}\bigl[u_0,u_1,u_2,u_{12}\bigr]$ 
of the differential function~$\boldsymbol{h}^{(k)}[\bu]$.
This is done by the procedure\footnote{The call is
\texttt{ToCoo(N,nf,nb,ex)}, where
\begin{itemize}
\item \texttt{N} is the number of Grassmann variables~$\theta_1,\ldots,\theta_N$;
\item \texttt{nf} is the number of fermionic super\/-\/fields \texttt{f(1)},$\ldots$,\texttt{f(nf)};
\item \texttt{nb} is the number of bosonic super\/-\/fields \texttt{b(1)},$\ldots$,\texttt{b(nb)};
\item \texttt{ex} is the super\/-\/field expression to be split in components.
\end{itemize}
For $N{=}2$, we have
\texttt{f(i)=f(i,0,0)+b(i,1,0)*th(1)+b(i,0,1)*th(2)+f(i,1,1)*th(1)*th(2)},
\texttt{b(i)=b(i,0,0)+f(i,1,0)*th(1)+f(i,0,1)*th(2)+b(i,1,1)*th(1)*th(2)}
as the splitting convention.
The reduction~\eqref{BRed} is achieved by setting
\texttt{b(i,0,1)}, \texttt{b(i,1,0)}, \texttt{f(j,0,1)}, 
and \texttt{f(j,1,0)} to zero for all $i\in[1,\mathtt{nb}]$ and~$j\in[1,\mathtt{nf}]$.}
\texttt{ToCoo},
which now is also available in~\textsc{SsTools}~\cite{SsTools}.
Thirdly, we set to zero the components~$u_1$ and~$u_2$ of the super\/-\/field~$\bu$. 
This gives the ansatz $h_{12}^{(k)}\bigl[u_0,u_{12}\bigr]$ for the representative of the conserved density in the vast equivalence class.
By the above, the gap between~$h_{12}^{(k)}$ and the known~$h^{(k)}$ amounts to~$\tfrac{\Id}{\Id x}\bigl(f^{(k)}\bigr)$, where $f^{(k)}\bigl[u_0,u_{12}\bigr]$ is a homogeneous differential polynomial. 
We remark that the choice of~$f$ is not unique due to the freedom in the choice of~$\boldsymbol{h}^{(k)}\mod\cD_1(\dots)+\cD_2(\dots)$.
We thus arrive at the linear algebraic equation
\begin{equation}\label{CutTail}
h_{12}^{(k)}-\tfrac{\Id}{\Id x}f^{(k)}=h^{(k)},
\end{equation}
which exprimes the equality of the respective coefficients in the polynomials.
The homogeneous polynomial ansatz for~$f^{(k)}$ is again generated 
by~\texttt{GenSSPoly}. Then equation~\eqref{CutTail} is split to the algebraic system by~\textsc{SsTools} and solved by \textsc{Crack}~\cite{WolfCrack}. Hence we obtain the coefficients in~$h_{12}^{(k)}$ and~$f^{(k)}$. \textit{A posteriori},
the freedom in the choice of~$f^{(k)}$ is redundant, and it is convenient to set the surviving \emph{unassigned} coefficients to zero. Indeed, they originate from the choice of a representative from the equivalence class for the super\/-\/density~$\boldsymbol{h}^{(k)}[\bu]$. This concludes the algorithm for the recursive production of homogeneous bosonic $N{=}2$ supersymmetry\/-\/invariant super\/-\/Hamiltonians~$\boldsymbol{\cH}^{(k)}$ for the $N{=}2$,\ $a{=}4$--\/SKdV hierarchy.

\begin{example}
Let us reproduce the first seven super\/-\/Hamiltonians for~\eqref{SKdV}, which were found in~\cite{MathieuNew}. 
In contrast with Example~\ref{ExBLim7dens},
we now list the \emph{properly chosen} 
representatives~$h^{(k)}_{12}\bigl[u_0,u_{12}\bigr]$ for the equivalence classes of conserved densities~$\tu_{0}^{(2k)}$ and $\tu_{12}^{(2k)}$, here~$k\leq3$.
Then we expose the conserved super\/-\/densities~$\boldsymbol{h}^{(k)}$ such that the respective expressions $h_{12}^{(k)}$ are obtained from the 
$\theta_1\theta_2$-\/components 
$\int\boldsymbol{h}^{(k)}\,\Id\boldsymbol{\theta}$ by the reduction~\eqref{BRed}.
\begin{subequations}\label{Hamsa}
\begin{align}
h_{12}^{(0)}&=u_0\sim\tu_0^{(0)},\qquad
\boldsymbol{h}^{(0)}=-\cD_1\cD_2(\bu)\sim 0,\\
h_{12}^{(1)}&=u_{12}\sim\tu_{12}^{(0)},\qquad
\boldsymbol{h}^{(1)}=\bu,\\
h_{12}^{(2)}&=-2u_{12}u_{0} \sim\tu_0^{(2)},\qquad
\boldsymbol{h}^{(2)}=\bu^2,\\
h_{12}^{(3)}&=\tfrac{3}{4}u_{12}^2 - 3u_{12}u_{0}^2 +\tfrac{3}{4} u_{0;x}^2 \sim \tu_{12}^{(2)},\qquad
\boldsymbol{h}^{(3)}=\bu^3 - \tfrac{3}{4}\bu\cD_1\cD_2(\bu),\\
h_{12}^{(4)}&= 3u_{12}^2u_{0} - 4u_{12}u_{0}^3 - \tfrac{3}{2}u_{0}^2u_{0;xx} - u_{12;x}u_{0;x} \sim\tu_0^{(4)},\notag\\
\boldsymbol{h}^{(4)}&=\bu^4 -\tfrac{1}{2}\bu \bu_{xx} - \tfrac{3}{2}\bu^2\cD_1\cD_2(\bu),\\
h_{12}^{(5)}&= -\tfrac{5}{4}u_{12}^3 + \tfrac{15}{2}u_{12}^2u_{0}^2 - 5u_{12}u_{0}^4 + 5u_{12}u_{0}u_{0;xx} + \tfrac{15}{8}u_{12}u_{0;x}^2 
+ \tfrac{15}{2}u_{0}^2u_{0;x}^2 + \tfrac{5}{16}u_{12;x}^2 {} + \notag\\
{}&{}+ \tfrac{5}{16}u_{0;xx}^2 \sim \tu_{12}^{(4)},\qquad
\boldsymbol{h}^{(5)}= \bu^5 - \tfrac{15}{16}\bu^2 \bu_{xx} + \tfrac{5}{8} (\cD_1\cD_2 \bu)^2 \bu - \tfrac{5}{2}\bu^3\cD_1\cD_2 \bu,\\
h_{12}^{(6)}&= -\tfrac{15}{4}u_{12}^3u_{0} + 15u_{12}^2u_{0}^3 
- \tfrac{15}{8}u_{12}^2u_{0;xx} - 6u_{12}u_{0}^5 - \tfrac{75}{4}u_{12}u_{0}u_{0;x}^2 - \tfrac{3}{8}u_{12}u_{0;xxxx} + {} \notag\\
{} & {}+ 5u_{0}^3u_{12;xx} + 15u_{0}^3u_{0;x}^2 + \tfrac{15}{8}u_{0}u_{12;x}^2 + \tfrac{15}{8}u_{0}u_{0;xx}^2 \sim \tu_0^{(6)},\notag
\end{align}\begin{align}%
\boldsymbol{h}^{(6)}&=\bu^6 - \tfrac{15}{8}\bu^3 \bu_{xx} 
+ \tfrac{3}{16}\bu \bu_{4x} + \tfrac{15}{8}(\cD_1\cD_2\bu)^2 - \tfrac{15}{4}\bu^4\cD_1\cD_2\bu + \tfrac{15}{8}\bu_{xx}\cD_1\cD_2\bu - {} \notag\\
{}& {}+ \tfrac{5}{8}\cD_1\cD_2(\bu)\cD_1(\bu)\cD_1(\bu_x),\\
%
h_{12}^{(7)}&= -\tfrac{21}{8}u_{0;4x}u_{0}u_{12} + \tfrac{7}{64}u_{0;xxx}^2 + \tfrac{105}{16}u_{0;xx}^2u_{0}^2 + \tfrac{35}{32}u_{0;xx}^2u_{12} 
- \tfrac{105}{8}u_{0;xx}u_{0}u_{12}^2 - \tfrac{105}{64}u_{0;4x}^4 - {}\notag\\
{}&{}- \tfrac{35}{16}u_{0;x}^2u_{12;xx} + \tfrac{105}{4}u_{0;x}^2u_{0}^4 
- \tfrac{525}{8}u_{0;x}^2u_{0}^2u_{12} - \tfrac{175}{32}u_{0;x}^2u_{12}^2 
+ \tfrac{7}{64}u_{12;xx}^2 + \tfrac{35}{4}u_{12;xx}u_{0}^4 + {} \notag\\
{}&{}+ \tfrac{105}{16}u_{12;x}^2u_{0}^2 - \tfrac{35}{32}u_{12;x}^2u_{12} 
- 7u_{0}^6u_{12} + \tfrac{105}{4}u_{0}^4u_{12}^2 - \tfrac{105}{8}u_{0}^2u_{12}^3 + \tfrac{35}{64}u_{12}^4 \sim \tu_{12}^{(6)},\notag\\
%
\boldsymbol{h}^{(7)}&= \bu^7 - \tfrac{105}{32}\bu^3 \bu_{xx} 
+ \tfrac{7}{32}\bu^2 \bu_{4x} - \tfrac{35}{64}\bu(\cD_1\cD_2\bu)^3 
+ \tfrac{35}{8}\bu^3(\cD_1\cD_2\bu)^2 
- \tfrac{35}{64}(\cD_1\cD_2\bu)^2 \bu_{xx} - {} \notag\\
{}& {} - \tfrac{21}{4}\bu^5\cD_1\cD_2\bu + \tfrac{105}{16}\bu^2\bu_{xx} \cD_1\cD_2\bu + \tfrac{315}{64}\bu\bu_{x}^2\cD_1\cD_2\bu 
+ \tfrac{35}{16}\bu(\cD_1\cD_2\bu)(\cD_1\bu)(\cD_1\bu_x) - {} \notag\\
{}& {} - \tfrac{7}{64}\bu_{4x}\cD_1\cD_2\bu - \tfrac{7}{8}\bu(\cD_1 \bu_{xx} )(\cD_1 \bu_x).
\end{align}
\end{subequations}
Of course, our super\/-\/densities~$\boldsymbol{h}^{(k)}$ are equivalent to those in~\cite{MathieuNew} up to trivial terms~$\cD_1(\dots)+\cD_2(\dots)$.
\end{example}

\begin{rem}
Until now, we have not yet reported 
any attempt of construction of Gardner's \emph{super\/-\/field} deformation
for~\eqref{SKdV}, which means that the ansatz for $\gm_\epsilon$ and~$\cE(\epsilon)$ is written in super\/-\/functions of~$\bu$ 
(c.f.~\cite{MathieuNew}).
This would yield the super\/-\/Hamiltonians~$\boldsymbol{\cH}^{(k)}$ at once,
and the intermediate deformation~\eqref{AppR} of a reduction~\eqref{BRed} for~\eqref{SKdV} would not be necessary. 
At the same time, the knowledge of Gardner's deformations for the reductions allows to inherit a part of the coefficients in the super\/-\/field ansatz by
fixing them in the component expansions 
(\textit{e.g.}, see~\eqref{DefKdV}, \eqref{BLimMiura}, and~\eqref{AppR}).
\end{rem}

Unfortunately, this cut\/-\/through does not work for the $N{=}2$,\ $a{=}4$--\/SKdV equation.

\begin{theor}[\textmd{$N{=}2$,\ $a{=}4$ `no go'}]\label{N2NoGo}
Under the assumptions that 
$N{=}2$ supersymmetry\/-\/invariant 
Gardner's deformation $\gm_\epsilon\colon\cE(\epsilon)\to\cE$ of~\eqref{SKdV} with $a{=}4$ be regular at~$\epsilon=0$, be scaling\/-homogeneous, and
retract to~\eqref{DefKdV} under the reduction $u_0=0$,\ $u_1=u_2=0$ in the super\/-\/field~\eqref{N2SuperField}, there is no such deformation.
\end{theor}

This rigidity statement, although under a principally different set of initial hypotheses, 
is contained in~\cite{MathieuNew}. In particular, there it was supposed that $\deg\gm_\epsilon=\deg\cE(\epsilon)=2$, which turns to be on the obstruction threshold, see below. We reveal the general nature of this `no go' result.

\begin{proof}
Suppose there is the super\/-\/field Miura contraction~$\gm_\epsilon$,
\begin{multline*}
\bu = \tilde{\bu}  
 +\epsilon\bigl(p_3\tilde{u}^2 - p_1\cD_1\cD_2\tilde{\bu} + p_2\tilde{\bu}_x\bigr) 
 +\epsilon^2\Bigl(p_{15}\tilde{\bu}^3 + p_{13}\tilde{\bu}\tilde{\bu}_x 
   + p_{10}\cD_2(\tilde{\bu})\cD_1(\tilde{\bu}) \\
   - p_{12}\cD_1\cD_2(\tilde{\bu})\tilde{\bu}  
   - p_{11}\cD_1\cD_2(\tilde{\bu}_x) + p_{14}\tilde{\bu}_{xx}\Bigr) + \cdots.
\end{multline*}
To recover the deformation~\eqref{DefKdV} upon~$u_{12}$ in~$\bu$, 
we split~$\gm_\epsilon$ in components and fix the coefficients of~$\epsilon\tu_{12;x}$ and~$\epsilon^2\tu_{12}^2$, see~\eqref{KdVeKdV}.
By this argument, the expansion of~$\tilde{\bu}_x$ yields $p_2=1$, while the equality
$-p_{12}\cD_1\cD_2(\tilde{\bu})\tilde{\bu} + p_{10}\cD_2(\tilde{\bu})\cD_1(\tilde{\bu})
=(p_{12}-p_{10})\theta_1\theta_2u_{12}^2+\dots$
implies that~$p_{12}=p_{10}-1$.
Next, we generate the homogeneous ansatz for~$\cE(\epsilon)$, which contains
$\tilde{\bu}_t=\dots+\epsilon^2\cdot\tfrac{\Id}{\Id x}\bigl(q_{17}(\cD_2\bu)(\cD_1\bu)\bu +\dots\bigr)+\dots$ in the right\/-\/hand side (the coefficient $q_{17}$ will appear in the obstruction).
We stress that now both~$\gm_\epsilon$ and $\cE(\epsilon)$ can be formal power series in~$\epsilon$ without any finite\/-\/degree polynomial truncation.

Now we split the determining equation~$\gm_\epsilon\colon\cE(\epsilon)\to\cE$
to the sequence of super\/-\/differential polynomial equalities ordered by the powers of~$\epsilon$. By the regularity assumption, the coefficients of higher powers of~$\epsilon$ never contribute to the equations that arise at its lower degrees. Consequently, every contradiction obtained at a finite order in the algebraic system is universal and precludes the existence of a solution. (Of course, we assume that the contradiction is not created artificially by an excessively low order polynomial truncation of the expansions in~$\epsilon$.)

This is the case for the $N{=}2$,\ $a{=}4$--\/SKdV. Using~\textsc{Crack}~\cite{WolfCrack}, we solve all but two algebraic equations in the quadratic approximation. The remaining system is
\[
q_{17}  = -p_{10}, \qquad   p_{10} + q_{17} + 1=0.
\]
This contradiction concludes the proof.
\end{proof}

\begin{rem}
In 
Theorem~\ref{N2NoGo} for~\eqref{SKdV} with~$a{=}4$, 
we state the non\/-\/existence of the Gardner deformation in a class of differential super\/-\/polynomials in~$\bu$, that is, of $N{=}2$ supersymmetry\/-\/invariant solutions that incorporate~\eqref{DefKdV}. 
Still, we do \emph{not} claim the non\/-\/existence of local regular Gardner's deformations for the four\/-\/component system~\eqref{SKdVComponents} in the class of 
differential functions of $u_0$, $u_1$, $u_2$, and~$u_{12}$.
\end{rem}

Consequently, it is worthy to deform the reductions of~\eqref{SKdVComponents} other than~\eqref{BRed}. Clearly, if there is a deformation for the entire system, then such partial solutions contribute to it by fixing the parts of the 
coefficients.

\begin{example}
Let us consider the reduction~$u_0=0$,\ $u_2=0$ in~\eqref{SKdVComponents} with~$a{=}4$. This is the two\/-\/component boson\/-\/fermion system
\begin{equation}\label{FB}
u_{1;t}=-u_{1;xxx}-3\bigl(u_1u_{12}\bigr)_x,\qquad
u_{12;t}=-u_{12;xxx}-6u_{12}u_{12;x}+3u_1u_{1;xx}.
\end{equation}
Notice that system~\eqref{FB} is \emph{quadratic}\/-\/nonlinear in both fields,
whence the balance $\deg\gm_\epsilon:\deg\cE(\epsilon)$ for its polynomial Gardner's deformations remains~$1:1$.

We found a unique Gardner's deformation of degree~$\leq4$ for~\eqref{FB}:
the Miura contraction~$\gm_\epsilon$ is cubic in~$\epsilon$,
\begin{subequations}\label{DefFB}
\begin{align}
u_1&=\tu_1,\qquad u_{12}=\tu_{12}
-\tfrac{1}{9} \epsilon^3 \tu_{1}\tu_{1;xx},\label{FBeFB}\\
\intertext{and the extension~$\cE(\epsilon)$ is given by the formulas}
\tu_{1;t} &=  - \tu_{1;xxx} - 3\bigl(\tu_1\tu_{12}\bigr)_x,\notag\\
\tu_{12;t} &= - \tu_{12;xxx} - 6\tu_{12}\tu_{12;x} + 3\tu_1\tu_{1;xx} + {}\notag \\
{}&{}\qquad{}+\tfrac{1}{3}\epsilon^3\Bigl( 
   u_1u_{1;xx}u_{12}-3u_1u_{1;x}u_{12;x}+u_{1;x}u_{1;xxx}
\Bigr)_x.\label{FBe}
\end{align}
\end{subequations}
However, we observe, first, that the contraction~\eqref{KdVeKdV} is not recovered\footnote{Surprisingly, the quadratic approximation~\eqref{KdVeKdV}
in the deformation problem for~\eqref{SKdVComponents} is very restrictive and leads to a unique solution~\eqref{BLimMiura}--\eqref{AppR} for~\eqref{BLim}.
Relaxing this constraint and thus permitting the coefficient of~$\epsilon^2\tu_{12}^2$ in~$\gm_\epsilon$ be arbitrary, 
we obtain two other real 
and two pairs of complex conjugate solutions 
for the deformations problem.        
They constitute the real and the complex orbit, respectively, 
under the action of the discrete symmetry 
$u_0\mapsto {-}u_0$, $\xi\mapsto {-}\xi$
of~\eqref{BLimBurg}.
}
by~\eqref{FBeFB} under~$u_1\equiv0$. Hence the deformation~\eqref{DefFB} and its mirror copy under~$u_1\leftrightarrow-u_2$ can not be merged
with~\eqref{BLimMiura} and~\eqref{AppR} to become parts of the deformation for~\eqref{SKdVComponents}.

Second, we recall that the fields~$u_1$ and~$u_2$ are, seemingly, the only local fermionic conserved densities for~\eqref{SKdVComponents} with~$a{=}4$.
Consequently, either the velocities~$\tu_{1;t}$ and~$\tu_{2;t}$ in  Gardner's extensions~$\cE(\epsilon)$ of~\eqref{SKdVComponents} are not expressed in the form of conserved currents 
(although this is indeed so at~$\epsilon=0$)
or the components $u_i=u_i\bigl(\bigl[\tu_0,\tu_1,\tu_2,\tu_{12}\bigr],\epsilon\bigr)$ of the Miura contractions~$\gm_\epsilon$ are the identity 
mappings~$u_i=\tu_i$, here~$i=1,2$, whence either the Taylor coefficients~$\tu_i^{(k)}$ of~$\tu_i$ are not termwise conserved on~\eqref{SKdVComponents} or
there appear no recurrence relations at all. This will be the object of another paper.
\end{example}

\section*{Conclusion}
\noindent%
We obtained the no\/-\/go statement 
for regular, scaling\/-\/homogeneous polynomial Gardner's deformations of the $N{=}2$,\ $a{=}4$--\/SKdV equation under the assumption that the solutions retract to the original formulas~\eqref{DefKdV} by Gardner~\cite{Gardner}.
At the same time, we found a new 
deformation~(\ref{BLimMiura}--\ref{BLimBurgE}) 
of the Kaup\/--\/Boussinesq equation~\eqref{BLimBurg}
that specifies the second flow in the bosonic limit of the super\/-\/hierarchy.
We emphasize that other known nontrivial deformations for 
the Kaup\/--\/Boussinesq equation~\cite{PamukKale} 
can be used for this purpose with equal success.

We exposed the two\/-\/step procedure
for recursive production of the 
bosonic super\/-\/Hamiltonians~$\boldsymbol{\cH}^{(k)}$.
We formulated the entire algorithm in full detail such that,
with elementary modifications, it is applicable to other 
supersymmetric KdV\/-\/type systems. 


\subsection*{Acknowledgements}
The authors thank P.~Mathieu, 
J.~W.~van de Leur, 
and Z.~Popowicz 
for helpful 
discussions.
This research is partially supported by NSERC (for V.\,H. and T.~W.)
and NWO grants~B61--609 and VENI~639.031.623
 (for A.\,V.\,K.). 
A part of this research was done while A.~V.~K. was visiting at
CRM (
Mont\-r\'eal) and Max Planck Institute for Mathematics (Bonn)%
, and A.~O.~K. was visiting at Utrecht University; 
the financial support and 
hospitality of these institutions are
gratefully acknowledged.
A.~V.~K. thanks 
the Organizing committee of 9th International workshop SQS'09 for support
and the organizers of 8th International conference `Symmetry in Nonlinear Mathematical Physics' for warm hospitality.

\end{document}